\documentclass[11pt]{article}

\usepackage{fullpage,amsthm}
\usepackage[disable]{todonotes}
\usepackage{graphicx} % support the \includegraphics command and options
\usepackage{array} % for better arrays (eg matrices) in maths
\usepackage{amsmath, amssymb, color, verbatim}
\usepackage{hyphenat,epsfig,subfigure,multirow}
\usepackage{hyperref}
\usepackage{bbold}
\usepackage{footnote}
\makesavenoteenv{table}
\makesavenoteenv{tabular}

\usepackage{algorithm2e}

\newtheorem{theorem}{Theorem}[section]
\newtheorem{lemma}[theorem]{Lemma}
\newtheorem{proposition}[theorem]{Proposition}
\newtheorem{corollary}[theorem]{Corollary}

\newtheorem{fact}[theorem]{Fact}
\newtheorem{conj}[theorem]{Conjecture}
\newtheorem{definition}[theorem]{Definition}
\newtheorem{remark}[theorem]{Remark}

\DeclareMathOperator*{\argmax}{arg\,max}

\renewcommand{\qed}{\nobreak \ifvmode \relax \else
      \ifdim\lastskip<1.5em \hskip-\lastskip
      \hskip1.5em plus0em minus0.5em \fi \nobreak
      \vrule height0.75em width0.5em depth0.25em\fi}

\newcommand{\eps}{\epsilon}

\def\E{\mathop{\mathbb{E}}\displaylimits}

\newcommand{\cA}{\mathcal A}
\newcommand{\cB}{\mathcal B}
\newcommand{\cS}{\mathcal S}

\newcommand{\ds}{D^{lin}}

\newcommand{\fr}{\text{freq }}

\newcommand{\R}[0]{{\ensuremath{\mathbb{R}}}}

\DeclareMathOperator{\argmin}{argmin}

\newcommand{\dl}{D^{lin}}
\newcommand{\rl}[1]{R^{lin}_{#1}}
\newcommand{\distl}[2]{\mathcal D^{lin, #2}_{#1}}
\newcommand{\distcm}[1]{\mathcal D^{\rightarrow}_{#1}}
\newcommand{\distlm}[1]{\mathcal D^{lin}_{#1}}
\newcommand{\distc}[2]{\mathcal D^{\rightarrow, #2}_{#1}}
\newcommand{\distlu}[1]{\distl{#1}{U}}
\newcommand{\distcu}[1]{\distc{#1}{U}}
\newcommand{\dc}{D^{\rightarrow}}
\newcommand{\rc}[1]{R^{\rightarrow}_{#1}}

\newcommand{\F}{\mathbb F}

\newcommand{\oo}{\{+1,-1\}}
\newcommand{\ftwo}{\F_2}
\def\E{\mathop{\mathbb{E}}\displaylimits}

\newcommand{\fplus}[1]{f^{+#1}}
\newcommand{\epsu}[2]{\eps_{#1}(#2)}
\newcommand{\dgap}[2]{\Delta_{#1}(#2)}

\newcommand{\mt}{Maj_3}
\newcommand{\mtk}[1]{\mt^{\circ #1}}

\title{Linear Sketching over $\mathbb F_2$}

\author{ Sampath Kannan \thanks{University of Pennsylvania, \texttt{kannan@cis.upenn.edu}}
	\and Elchanan Mossel \thanks{Massachusetts Institute of Technology, \texttt{elmos@mit.edu}. E.M. acknowledges the support of grant N00014-16-1-2227 from Office of Naval Research and of NSF award CCF 1320105 as well as support from Simons Think Tank on Geometry \& Algorithms.}
	\and Grigory Yaroslavtsev \thanks{Indiana University, Bloomington \texttt{grigory@grigory.us}}
	}

\begin{document}

\listoftodos

%\section*{Notation}
%\begin{itemize}
%\item $\dl$ = Deterministic linear sketch complexity.
%\item $\rl{\delta}$ = Randomized linear sketch complexity with error $\delta$.
%\item $\distl{\delta}{U}$ = Distributional linear sketch complexity over uniform distribution.
%\item $\distlm{\delta}$ = Distributional linear sketch complexity.
%\item $\dc$ = Deterministic one-way communication complexity.
%\item $\rc{\delta}$ = Randomized one-way communication complexity with error $\delta$.
%\item $\distc{\delta}{U}$ = Distributional one-way communication complexity over uniform distribution.
%\item $\distcm{\delta}$ = Distributional one-way communication complexity.
%\end{itemize}

\newpage

\maketitle
\thispagestyle{empty}
\begin{abstract}
	We initiate a systematic study of linear sketching over $\ftwo$.
	For a given Boolean function $f \colon \{0,1\}^n \to \{0,1\}$ a randomized $\ftwo$-sketch is a distribution $\mathcal M$ over $d \times n$ matrices with elements over $\ftwo$ such that $\mathcal Mx$ suffices for computing $f(x)$ with high probability. We study a connection between $\ftwo$-sketching and a two-player one-way communication game for the corresponding XOR-function. Our results show that this communication game characterizes $\ftwo$-sketching under the uniform distribution (up to dependence on error).
	Implications of this result include: 1) a composition theorem for $\ftwo$-sketching complexity of a recursive majority function, 2) a tight relationship between $\ftwo$-sketching complexity and Fourier sparsity, 3) lower bounds for a certain subclass of symmetric functions.
	We also fully resolve a conjecture of Montanaro and Osborne regarding one-way communication complexity of linear threshold functions by designing an $\ftwo$-sketch of optimal size.
	
	Furthermore, we show that (non-uniform) \todo{add some discussion of non-uniformity?} streaming algorithms that have to process random updates over $\ftwo$ can be constructed as $\ftwo$-sketches for the uniform distribution with only a minor loss. In contrast with the previous work of Li, Nguyen and Woodruff (STOC'14) who show an analogous result for linear sketches over integers in the adversarial setting our result doesn't require the stream length to be triply exponential in $n$ and holds for streams of length $\tilde O(n)$ constructed through uniformly random updates.
	Finally, we state a conjecture that asks whether optimal one-way communication protocols for XOR-functions can be constructed as $\ftwo$-sketches with only a small loss.
\end{abstract}

\newpage
\thispagestyle{empty}
\tableofcontents

\newpage

\setcounter{page}{1}

\section{Introduction}

Linear sketching is the underlying technique behind many of the biggest algorithmic breakthroughs of the past two decades. 
It has played a key role in the development of streaming algorithms since~\cite{AMS99}\todo{more citations} and most recently has been the key to modern randomized algorithms for numerical linear algebra (see survey~\cite{W14}), graph compression (see survey~\cite{M14}), dimensionality reduction\todo{cite some of the Jelani's papers}, etc.
Linear sketching is robust to the choice of a computational model and can be applied in settings as seemingly diverse as streaming, MapReduce as well as various other distributed models of computation~\cite{HPPSS15}, allowing to save computational time, space and reduce communication in distributed settings. 
This remarkable versatility is based on properties of linear sketches enabled by linearity: simple and fast updates and mergeability of sketches computed on distributed data. Compatibility with fast numerical linear algebra packages makes linear sketching particularly attractive for applications.

Even more surprisingly linear sketching over the reals is known to be the best possible algorithmic approach (unconditionally) in certain settings. Most notably, under some mild conditions linear sketches are known to be almost space optimal for processing dynamic data streams~\cite{G08,LNW14,AHLW16}.
Optimal bounds for streaming algorithms for a variety of computational problems can be derived through this connection by analyzing linear sketches rather than general algorithms. Examples include approximate matchings~\cite{AKLY16}, additive norm approximation~\cite{AHLW16} and frequency moments~\cite{LNW14}. 

In this paper we study the power of linear sketching over $\ftwo$.
\footnote{It is easy to see that sketching over finite fields can be significantly better than linear sketching over integers for certain computations. As an example, consider a function $(x \mod 2)$ (for an integer input $x$) which can be trivially sketched with $1$ bit over the field of two elements while any linear sketch over the integers requires word-size memory.}
To the best of our knowledge no such systematic study currently exists as prior work focuses on sketching over the field of reals (or large finite fields as reals are represented as word-size bounded integers).
Formally, given a function $f \colon \{0,1\}^n \rightarrow \{0,1\}$ that needs to be evaluated over an input $x = (x_1, \dots, x_n)$ we are looking for a distribution over $k$ subsets $\mathbf{S}_1, \dots, \mathbf{S}_k \subseteq [n]$ such that the following holds: for any input $x$ given parities computed over these sets and denoted as $\chi_{\mathbf{S}_1}(x), \chi_{\mathbf{S}_2}(x), \dots, \chi_{\mathbf{S}_k}(x)$\footnote{Here we use notation $\chi_S(x) = \oplus_{i \in S} x_i$.} it should be possible to compute $f(x)$ with  probability $ 1- \delta$.
In the matrix form sketching corresponds to multiplication over $\mathbb F_2$ of the row vector $x$ by a random $n \times k$ matrix  whose $i$-th column is a characteristic vector of the random parity $\chi_{\mathbf{S}_i}$:
\vspace{-15pt}
$$
\bordermatrix {
	&&&&\cr
	& x_1     & x_2     & \ldots & x_n  \cr
}
\bordermatrix{
	&     &      &  &      \cr
	& \vdots & \vdots & \vdots & \vdots     \cr
	& \chi_{\mathbf{S}_1}     & \chi_{\mathbf{S}_2}     & \ldots & \chi_{\mathbf{S}_k}     \cr
	& \vdots & \vdots & \vdots & \vdots \cr
}
=\bordermatrix {
	&&&&\cr
	& \chi_{\mathbf{S}_1}(x)     & \chi_{\mathbf{S}_2}(x)     & \ldots & \chi_{\mathbf{S}_k}(x)  \cr
}
$$
This sketch alone should then be sufficient for computing $f$ with high probability for any input $x$. 
This motivates us to define the \textit{randomized linear sketch} complexity of a function $f$ over $\ftwo$  as the smallest $k$ which allows to satisfy the above guarantee.\todo{switch notation here}
\begin{definition}[$\ftwo$-sketching]\label{def:rand-f2-sketch}
For a function $f \colon \ftwo^n \to \ftwo$ we define its \emph{randomized linear sketch complexity}\todo{maybe replace linear sketch with $\ftwo$-sketch everywhere?}
\footnote{In the language of decision trees this can be interpreted as randomized non-adaptive parity decision tree complexity. We are unaware of any systematic study of this quantity either. Since heavy decision tree terminology seems excessive for our applications (in particular, sketching is done in one shot so there isn't a decision tree involved) we prefer to use a shorter and more descriptive name.} over $\ftwo$ with error $\delta$ (denoted as $\rl{\delta}(f)$) as the smallest integer $k$ such that there exists a distribution $\chi_{\mathbf{S}_1},\chi_{\mathbf{S}_2},\ldots, \chi_{\mathbf{S}_k}$ over $k$ linear functions over $\ftwo$ and a postprocessing function $g:\ftwo^k \rightarrow \ftwo$\footnote{If a random family of functions is used here then the definition is changed accordingly. In this paper all $g$ are deterministic.} which satisfies:
$$
\forall x \in \ftwo^n \colon \Pr_{\mathbf{S}_1, \dots, \mathbf{S}_k}[f(x_1,x_2,\ldots,x_n) =
g(\chi_{\mathbf{S}_1}(x),\chi_{\mathbf{S}_2}(x),\ldots, \chi_{\mathbf{S}_k}(x))] \ge 1-\delta.
$$
\end{definition}

As we show in this paper the study of $\rl{\delta}(f)$ is closely related to a certain communication complexity problem.
For $f \colon \ftwo^n \to \ftwo$ define the XOR-function $\fplus{}\colon \ftwo^n \times \ftwo^n \to \ftwo$ as $\fplus{}(x,y) = f(x + y)$ where $x,y \in \ftwo^n$.
Consider a communication game between two players Alice and Bob holding inputs $x$ and $y$ respectively.
Given access to a shared source of random bits Alice has to send a single message to Bob so that he can compute $\fplus{}(x,y)$.
This is known as the one-way communication complexity problem for XOR-functions. 
\begin{definition}[Randomized one-way communication complexity of XOR function]
For a function $f \colon \ftwo^n \to \ftwo$ the \emph{randomized one-way communication complexity} with error $\delta$ (denoted as $\rc{\delta}(\fplus{})$) of its XOR-function is defined as the smallest size\footnote{Formally the minimum here is taken over all possible protocols where for each protocol the size of the message $M(x)$ refers to the largest size (in bits) of such message taken over all inputs $x \in \ftwo^n$. See~\cite{KN97} for a formal definition.} (in bits) of the (randomized using public randomness) message $M(x)$ from Alice to Bob which allows Bob to evaluate $\fplus{}(x,y)$ for any $x,y \in \ftwo^n$ with error probability at most $\delta$.
\end{definition}
Communication complexity complexity of XOR-functions has been recently studied extensively in the context of the log-rank conjecture (see e.g.~\cite{SZ08,ZS10,MO09,LZ10,LLZ11,SW12,LZ13,TWXZ13,L14,HHL16}). However, such studies either mostly focus on deterministic communication complexity or are specific to the two-way communication model. We discuss implications of this line of work for our $\ftwo$-sketching model in our discussion of prior work.

It is easy to see that $\rc{\delta}(\fplus{}) \le \rl{\delta}(f)$ as using shared randomness Alice can just  send $k$ bits $\chi_{\mathbf{S}_1}(x),\chi_{\mathbf{S}_2}(x),\ldots, \chi_{\mathbf{S}_k}(x)$ to Bob
who can for each $i \in [k]$ compute $\chi_{\mathbf{S}_i}(x + y) = \chi_{\mathbf{S}_i}(x) + \chi_{\mathbf{S}_i}(y)$, which is an $\ftwo$-sketch of $f$ on $x + y$ and hence suffices for computing $\fplus{}(x,y)$ with probability $1 - \delta$. The main open question raised in our work is whether the reverse inequality holds (at least approximately), thus implying the equivalence of the two notions.
\begin{conj}\label{conj:main}
Is it true that $\rc{\delta}(\fplus{}) = \tilde \Theta\left(\rl{\delta}(f)\right)$ for every $f \colon \ftwo^n \to \ftwo$ and $0 < \delta < 1/2$?
\end{conj}
In fact all known one-way protocols for XOR-functions can be seen as $\ftwo$-sketches so it is natural to ask whether this is always true. In this paper we further motivate this conjecture through a number of examples of classes of functions for which it holds.
One important such example from the previous work is a function $Ham_{\ge k}$ which evaluates to $1$ if and only if the Hamming weight of the input string is at least $k$. The corresponding XOR-function $Ham_{\ge k}^+$ can be seen to have one-way communication complexity of $\Theta(k \log k)$ via the small set disjointness lower bound of~\cite{DKS12} and a basic upper bound based on random parities~\cite{HSZZ06}.
Conjecture~\ref{conj:main} would imply that in order to prove a one-way disjointness lower bound it suffices to only consider $\ftwo$-sketches.

In the discussion below using Yao's principle we switch to the equivalent notion of distributional complexity of the above problems denoted as $\distcm{\delta}$ and $\distlm{\delta}$ respectively.
For the formal definitions we refer to the reader to Section~\ref{sec:communication-complexity} and a standard textbook on communication complexity~\cite{KN97}.
Equivalence between randomized and distributional complexities allows us to restate Conjecture~\ref{conj:main} as $\distcm{\delta} = \tilde{\Theta}(\distlm{\delta})$.

For a fixed distribution $\mu$ over $\ftwo^n$ we define $\distl{\delta}{\mu}(f)$ to be the smallest dimension of an $\ftwo$-sketch that correctly outputs $f$ with probability $1- \delta$ over $\mu$.
Similarly for a distribution $\mu$ over $(x,y)\in \ftwo^n\times \ftwo^n$ we denote distributional one-way communication complexity of $f$ with error $\delta$ as $\distc{\delta}{\mu}(\fplus{})$ (See Section~\ref{sec:prelims} for a formal definition). Our first main result is an analog of Conjecture~\ref{conj:main} for the uniform distribution $U$ over $(x,y)$ that matches the statement of the conjecture up to dependence on the error probability: 

\begin{theorem}\label{thm:linear-sketch-uniform-main}
	For any $f \colon \ftwo^n \to \ftwo$ it holds that
	$\distcu{\Theta(\frac1n)}(\fplus{}) \ge \distlu{\frac13}(f)$.
\end{theorem}

A deterministic analog of Definition~\ref{def:rand-f2-sketch} requires that $f(x) = g(\chi_{\alpha_1}(x),\chi_{\alpha_2}(x),\ldots, \chi_{\alpha_k}(x))$ for a fixed choice of $\alpha_1, \dots, \alpha_k \in \ftwo^n$.
The smallest value of $k$ which satisfies this definition is known to be equal to the Fourier dimension of $f$ denoted as $dim(f)$. It corresponds to the smallest dimension of a linear subspace of $\ftwo^n$ that contains the entire spectrum of $f$ (see Section~\ref{sec:fourier} for a formal definition).
In order to keep the notation uniform we also denote it as $\dl(f)$.
Most importantly, as shown in~\cite{MO09} an analog of Conjecture~\ref{conj:main} holds without any loss in the deterministic case, i.e. $\dc(\fplus{}) = dim(f) = \dl(f)$, where $\dc$ denotes the deterministic one-way communication complexity.
This striking fact is one of the reasons why we suggest Conjecture~\ref{conj:main} as an open problem.

In order to prove Theorem~\ref{thm:linear-sketch-uniform-main} we introduce a notion of an \textit{approximate Fourier dimension} (Definition~\ref{def:approx-fourier-dim}) that extends the definition of exact Fourier dimension to allow that only $1 - \epsilon$ fraction of the total ``energy'' in $f$'s spectrum should be contained in the linear subspace.
The key ingredient in the proof is a structural theorem Theorem~\ref{thm:linear-sketch-uniform} that characterizes both $\distlu{\delta}(f)$ and $\distcu{\delta}(\fplus{})$ in terms of $f$'s approximate Fourier dimension.

\subsection*{Previous work and our results}
Using Theorem~\ref{thm:linear-sketch-uniform} we derive a number of results that confirm Conjecture~\ref{conj:main} for specific classes of functions.

\paragraph{Recursive majority}
For an odd integer $n$ the majority function $Maj_n$ is defined as to be equal $1$ if and only if the Hamming weight of the input is greater than $n/2$. Of particular interest is the recursive majority function $Maj_3^{\circ k}$ that corresponds to $k$-fold composition of $Maj_3$ for $k = \log_3 n$.
This function was introduced by Boppana~\cite{SW86} and serves as an important example of various properties of Boolean functions, most importantly in randomized decision tree complexity ~(\cite{SW86,JKS03,MNSX11,L13,MNSSTX13}) and most recently deterministic parity decision tree complexity~\cite{BTW15}.

In Section~\ref{sec:rec-majority} we show to use Theorem~\ref{thm:linear-sketch-uniform} to obtain the following result:
\begin{theorem}\label{thm:rec-majority}
	For any $\epsilon \in [0,1]$, $\gamma < \frac12 - \epsilon$ and $k = \log_3 n$ it holds that: $$\distcu{\frac1n \left(\frac14 - \eps^2\right)}({\mtk{k}}^+) \ge \epsilon^2 n + 1.$$ 
\end{theorem}
In particular, this confirms Conjecture~\ref{conj:main} for $Maj_3^{\circ k}$ with at most logarithmic gap as for constant $\epsilon$ we get $\distcu{\Theta(\frac1n)}({\mtk{k}}^+) = \Omega(n)$.
By Yao's principle $\rc{\Theta(\frac1n)}({\mtk{k}}^+) = \Omega(n)$.
Using standard error reduction ~\cite{KN97} for randomized communication this implies that $\rc{\delta}({\mtk{k}}^+) = \tilde \Omega(n)$ for constant $\delta < 1/2$ almost matching the trivial upper bound.

\paragraph{Address function and Fourier sparsity}
The number $s$ of non-zero Fourier coefficients of $f$ (known as Fourier sparsity) is one of the key quantities in the analysis of Boolean functions. It also plays an important role in the recent work on log-rank conjecture for XOR-functions~\cite{TWXZ13,STV14}. 
A remarkable recent result by Sanyal~\cite{S15} shows that for Boolean functions $dim(f) = O(\sqrt{s} \log s)$, namely all non-zero Fourier coefficients are contained in a subspace of a polynomially smaller dimension.
This bound is almost tight as the \textit{address function} (see Section~\ref{sec:address} for a definition) exhibits a quadratic gap.
A direct implication of Sanyal's result is a deterministic $\ftwo$-sketching upper bound of $O(\sqrt{s}\log s)$ for any $f$ with Fourier sparsity $s$.
As we show in Section~\ref{sec:address} this dependence on sparsity can't be improved even if randomization is allowed.

\paragraph{Symmetric functions}
A function $f$ is symmetric if it only depends on the Hamming weight of its input.
In Section~\ref{sec:symmetric} we show that Conjecture~\ref{conj:main} holds (approximately) for symmetric functions which are not too close to a constant function or the parity function $\sum_i x_i$ where the sum is taken over $\ftwo$.

\paragraph{Applications to streaming}
In the turnstile streaming model of computation an vector $x$ of dimension $n$ is updated through a sequence of additive updates applied to its coordinates and the goal of the algorithm is to be able to output $f(x)$ at any point during the stream while using space that is sublinear in $n$.
In the real-valued case we have either $x \in [0,m]^n$ or $x \in [-m,m]^n$ for some universal upper bound $m$ and updates can be increments or decrements to $x$'s coordinates of arbitrary magnitude.

For $x \in \ftwo^n$ additive updates have a particularly simple form as they always flip the corresponding coordinate of $x$.
As we show in Section~\ref{sec:adversarial-streaming} it is easy to see based on the recent work of~\cite{G08,LNW14,AHLW16} that in the adversarial streaming setting the space complexity of turnstile streaming algorithms over $\ftwo$ is determined by the $\ftwo$-sketch complexity of the function of interest.
However, this proof technique only works for very long streams which are unrealistic in practice -- the length of the adversarial stream has to be triply exponential in $n$ in order to enforce linear behavior.
Large stream length requirement is inherent in the proof structure in this line of work and while one might expect to improve triply exponential dependence on $n$ at least an exponential dependence appears necessary, which is a major limitation of this approach.

As we show in Section~\ref{sec:random-streaming} it follows directly from our Theorem~\ref{thm:linear-sketch-uniform-main} that turnstile streaming algorithms that achieve low error probability under random $\ftwo$ updates might as well be $\ftwo$-sketches.
For two natural choices of the random update model short streams of length either $O(n)$ or $O(n \log n)$ suffice for our reduction.
We stress that our lower bounds are also stronger than the worst-case adversarial lower bounds as they hold under an average-case scenario.
Furthermore, our Conjecture~\ref{conj:main} would imply that space optimal turnstile streaming algorithms over $\ftwo$ have to be linear sketches for adversarial streams of length only $2n$.

\paragraph{Linear Threshold Functions}

Linear threshold functions (LTFs) are one of the most studied classes of Boolean functions as they play a central role in circuit complexity, learning theory and machine learning (See Chapter 5 in ~\cite{OD14} for a comprehensive introduction to properties of LTFs).
Such functions are parameterized by two parameters $\theta$ and $m$ known as threshold and margin respectively (See Definition~\ref{def:ltf} for a formal definition).
We design an $\ftwo$-sketch for LTFs with complexity $O(\theta /m \log (\theta/m))$.
By applying the sketch in the one-way communication setting this fully resolves an open problem posed in~\cite{MO09}.
Our work shows that dependence on $n$ is not necessary which is an improvement over previously best known protocol due to~\cite{LZ13} which achieves communication $O(\theta/m \log n)$.
Our communication bound is optimal due to~\cite{DKS12}. See Section~\ref{sec:ltf} for details.

\paragraph{Other previous work}
Closely related to ours is work on communication protocols for XOR-functions started in~\cite{SZ08,MO09}.
In particular~\cite{MO09} presents two basic one-way communication protocols based on random parities.
First one, stated as Fact~\ref{prop:l0-bound} generalizes the classic protocol for equality.
Second one uses the result of Grolmusz~\cite{G97} and implies that $\ell_1$-sampling of Fourier characters gives a randomized $\ftwo$-sketch of size $O(\|\hat f\|^2_1)$ (for constant error).
Another line of work that is closely related to ours is the study of the two-player simultaneous message passing model (SMP). This model can also allow to prove lower bounds on $\ftwo$-sketching complexity.
However, in the context of our work there is no substantial difference as for product distributions the two models are essentially equivalent.
Recent results in the SMP model include~\cite{MO09,LLZ11,LZ13}.

While decision tree literature is not directly relevant to us since our model doesn't allow adaptivity we remark that there has been interest recently in the study of (adaptive) deterministic parity decision trees~\cite{BTW15} and non-adaptive deterministic parity decision trees~\cite{STV14,S15}.
As mentioned above, our model can be interpreted as non-adaptive randomized parity decision trees and to the best of our knowledge it hasn't been studied explicitly before.
Another related model is that of \textit{parity kill numbers}. In this model a composition theorem has recently been shown by~\cite{OWZST14} but the key difference is again adaptivity.
\todo{cite all these recent lifting theorems?}

\paragraph{Organization}
The rest of this paper is organized as follows.
In Section~\ref{sec:prelims} we introduce the required background from communication complexity and Fourier analysis of Boolean functions.
In Section~\ref{sec:uniform} we prove Theorem~\ref{thm:linear-sketch-uniform-main}.
In Section~\ref{sec:applications} we give applications of this theorem for recursive majority (Theorem~\ref{thm:rec-majority}), address function and symmetric functions.
In Section~\ref{sec:streaming} we describe applications to streaming.
In Section~\ref{sec:ltf} we describe our $\ftwo$-sketching protocol for LTFs.
In Section~\ref{sec:one-bit-lb} we show a lower bound for one-bit protocols making progress towards resolving Conjecture~\ref{conj:main}.

In Appendix~\ref{app:deterministic} we give some basic results about deterministic $\ftwo$-sketching (or Fourier dimension) of composition and convolution of functions. We also present a basic lower bound argument based on affine dispersers.
In Appendix~\ref{app:randomized} we give some basic results about randomized $\ftwo$-sketching including a lower bound based on extractors and a classic protocol based on random parities which we use as a building block in our sketch for LTFs. We also present evidence for why an analog of Theorem~\ref{thm:linear-sketch-uniform} doesn't hold for arbitrary distributions.
In Appendix~\ref{app:tightness} we argue that the parameters of Theorem~\ref{thm:linear-sketch-uniform} can't be substantially improved.

\section{Preliminaries}\label{sec:prelims}

For an integer $n$ we use notation $[n] = \{1, \dots, n\}$. For integers $n \le m$ we use notation $[n,m] = \{n, \dots, m\}$.
For an arbitrary domain $\mathcal D$ we denote the uniform distribution over this domain as $U(\mathcal D)$.
For a vector $x$ and $p \ge 1$ we denote the $p$-norm of $x$ as $\|x\|_p$ and reserve the notation $\|x\|_0$ for the Hamming weight.

\subsection{Communication complexity}\label{sec:communication-complexity}
Consider a function $f \colon \ftwo^n \times \ftwo^n \to \ftwo$ and a distribution $\mu$ over $\ftwo^n \times \ftwo^n$.
The \textit{one-way distributional complexity} of $f$ with respect to $\mu$, denoted as $\distc{\delta}{\mu}(f)$ is the smallest communication cost of a one-way deterministic protocol that outputs $f(x,y)$ with probability at least $1 - \delta$ over the inputs $(x,y)$ drawn from the distribution $\mu$.
The \textit{one-way distributional complexity} of $f$ denoted as $\distcm{\delta}(f)$ is defined as $\distcm{\delta}(f) = \sup_\mu \distc{\delta}{\mu}(f)$.
By Yao's minimax theorem~\cite{Y83} it follows that $\rc{\delta}(f) = \distcm{\delta}(f)$.
\textit{One-way communication complexity over product distributions} is defined as $\distc{\delta}{\times}(f) = \sup_{\mu = \mu_x \times \mu_y} \distc{\delta}{\mu}(f)$ where $\mu_x$ and $\mu_y$ are distributions over $\ftwo^n$.

With every two-party function $f \colon \ftwo^n \times \ftwo^n$ we associate with it the \textit{communication matrix} $M^f \in \ftwo^{2^n \times 2^n}$ with entries $M^f_{x,y} = f(x,y)$.
We say that a deterministic protocol $M(x)$ with length $t$ of the message that Alice sends to Bob partitions the rows of this matrix into $2^t$ \textit{combinatorial rectangles} where each rectangle contains all rows of $M^f$ corresponding to the same fixed message $y \in \{0,1\}^t$.

\subsection{Fourier analysis}\label{sec:fourier}
We consider functions from $\ftwo^n$ to 
$\R$\footnote{
In all Fourier-analytic arguments Boolean functions are treated as functions of the form $f : \ftwo^n \to \oo$ where $0$ is mapped to $1$ and $1$ is mapped to $-1$. Otherwise we use these two notations interchangeably.}.
For any fixed $n \ge 1$, the space of these functions forms an inner product space
with the inner product
$\left<f, g\right> = \E_{x \in \ftwo^n}[ f(x) g(x) ] = \frac1{2^n} \sum_{x \in \ftwo^n} f(x)g(x)$.
The $\ell_2$ norm of $f : \ftwo^n \to \R$ is
$\| f \|_2 = \sqrt{ \left< f, f \right>} = \sqrt{\E_x[ f(x)^2 ]}$
and the $\ell_2$ distance between two functions $f, g : \ftwo^n \to \R$ is 
the $\ell_2$ norm of the function $f - g$.
In other words, $\|f - g \|_2 = \sqrt{\left< f-g, f-g \right>} 
= \frac1{\sqrt{|\ftwo^n|}} \sqrt{\sum_{x \in \ftwo^n} (f(x) - g(x))^2}$. \todo{Do we need the distance?}

%We let {\em Hamming distance} of a pair of functions $\phi,\psi$ with a common domain $\mathcal{D}$ as $\distHam(\phi,\psi)=\Pr_{x \in \mathcal{D}}[\phi(x) \neq \psi(x)]$. 

For $x,y\in \ftwo^n$ we denote the inner product as $x \cdot y = \sum_{i=1}^n x_i y_i$.
For $\alpha \in \ftwo^n$, the \emph{character} 
$\chi_\alpha : \ftwo^n \to \oo$ is the function defined by
$
\chi_\alpha(x) = (-1)^{\alpha \cdot x}.
$
Characters form an orthonormal basis as $\langle \chi_\alpha, \chi_\beta \rangle = \delta_{\alpha\beta}$ where $\delta$ is the Kronecker symbol.
The \emph{Fourier coefficient} of $f : \ftwo^n \to \R$ corresponding to $\alpha$ is
$
\hat{f}(\alpha) = \E_x[ f(x) \chi_\alpha(x)].
$
The \emph{Fourier transform} of $f$ is the function $\hat{f} : \ftwo^n \to \R$
that returns the value of each Fourier coefficient of $f$. 
We use notation $Spec(f) = \{\alpha \in \ftwo^n : \hat f(\alpha) \neq 0\}$ to denote the set of all non-zero Fourier coefficients of $f$.

The set of Fourier transforms of functions mapping $\ftwo^n \to \R$ forms an inner
product space with inner product
$
\left< \hat{f}, \hat{g} \right> = \sum_{\alpha \in \ftwo^n} \hat{f}(\alpha) \hat{g}(\alpha).
$
The corresponding $\ell_2$ norm is
$
\| \hat{f} \|_2 = \sqrt{\left< \hat{f}, \hat{f}\right>} = 
\sqrt{\sum_{\alpha \in \ftwo^n} \hat{f}(\alpha)^2}.
$
Note that the inner product and $\ell_2$ norm are weighted differently for a function $f : \ftwo^n \to \R$ and its Fourier transform $\hat{f} : \ftwo^n \to \R$.

\begin{fact}[Parseval's identity]
	For any $f : \ftwo^n \to \R$ it holds that
	$
	\|f \|_2 = \| \hat{f} \|_2 
	= \sqrt{ \sum_{\alpha \in \ftwo^n} \hat{f}(\alpha)^2 }.
	$
	Moreover, if $f : \ftwo^n \to \oo$ then $\|f\|_2 = \|\hat f\|_2 = 1$.
\end{fact}

We use notation $A \le \ftwo^n$ to denote the fact that $A$ is a linear subspace of $\ftwo^n$.
\begin{definition}[Fourier dimension]
The \emph{Fourier dimension} of $f \colon \ftwo^n \to \oo$ denoted as $dim(f)$ is the smallest integer $k$ such that there exists $A \le \ftwo^n$ of dimension $k$ for which $Spec(f) \subseteq A$.
\end{definition}
We say that $A \le \mathbb F_2^n$ is a \textit{standard subspace} if it has a basis $v_1, \dots, v_d$ where each $v_i$ has Hamming weight equal to $1$.
An \textit{orthogonal subspace} $A^{\perp}$ is defined as:
\begin{align*}
A^\perp = \{\gamma \in \ftwo^n : \forall x\in A \quad \gamma \cdot x = 0\}.
\end{align*}
An \textit{affine subspace} (or coset) of $\ftwo^n$ of the form $A = H + a$ for some $H \le \ftwo^n$ and $a \in \ftwo^n$ is defined as:
\begin{align*}
A = \{\gamma \in \ftwo^n : \forall x\in H^\perp \quad \gamma \cdot x = a \cdot x\}.
\end{align*}

We now introduce notation for restrictions of functions to affine subspaces.
\begin{definition}
Let $f : \ftwo^n \to \R$ and $z \in \ftwo^n$. We define $\fplus{z}: \ftwo^n \to \mathbb R$ as $\fplus{z}(x) = f(x + z)$.
\end{definition}
\begin{fact}\label{fact:fourier-shift}
Fourier coefficients of $\fplus{z}$ are given as $\widehat{\fplus{z}}(\gamma) = (-1)^{\gamma \cdot z} \hat f(\gamma)$ and hence:
\begin{align*}
\fplus{z}  = \sum_{S \in \ftwo^n} \hat f(S) \chi_S(z) \chi_S.
\end{align*}
\end{fact}
\begin{definition}[Coset restriction]
For $f : \ftwo^n \to \R, z \in \ftwo^n$ and $ H \le \ftwo^n$ we write $\fplus{z}_H \colon H \to \R$ for the restriction of $f$ to $H + z$.
\end{definition}\todo{Do we need this? Make sure we use it if so.}

\begin{definition} [Convolution]
	For two functions $f,g \colon \ftwo^n \to \mathbb R$ their convolution $(f * g) \colon \ftwo^n\rightarrow \mathbb R$ is defined as $(f * g)(x) = \mathbb E_{y \sim U(\ftwo^n)} \left[f(x) g(x + y)\right]$. 
\end{definition}
For $S \in \ftwo^n$ the corresponding Fourier coefficient of convolution is given as $\widehat {f * g}(S) = \hat f(S) \hat g(S)$.

%\section{Randomized Linear Sketching over G$\mathbb F_2$}

\section{$\mathbb F_2$-sketching over the uniform distribution}\label{sec:uniform}
We use the following definition of Fourier concentration that plays an important role in learning theory~\cite{KM93}.
\begin{definition}[Fourier concentration]
	The spectrum of a function $f:\ftwo^n \to \oo$ is $\eps$-concentrated on a collection of Fourier coefficients $Z \subseteq \ftwo^n$ if $\sum_{S \in Z} \hat f^2(S) \ge \eps.$
\end{definition}

For a function $f \colon \ftwo^n \to \oo$ and a parameter $\epsilon>0$ we introduce a notion of \emph{approximate Fourier dimension} as the smallest integer for which $f$ is $\epsilon$-concentrated on some linear subspace of dimension $d$.\todo{propagate off by one shift everywhere}

\begin{definition}[Approximate Fourier dimension]\label{def:approx-fourier-dim}
	Let $\mathcal A_k$ be the set of all linear subspaces of $\mathbb F_2^n$ of dimension $k$.
	For $f \colon \ftwo^n \to \oo$ and $\epsilon > 0$ the approximate Fourier dimension $dim_{\epsilon}(f)$ is defined as:
	$$dim_\epsilon(f) = \min_k \left\{ \exists A \in \mathcal A_k \colon \sum_{S \in A} \hat f(S)^2 \ge \epsilon\right\}.$$
\end{definition}

\begin{definition}[Approximate Fourier dimension gap]
	For $f \colon \ftwo^n \to \oo$ and $1 \le d \le n$ we define:
	\begin{align*}
%	\epsl{f} = \min_\eps \left\{dim_\eps(f) \ge d\right\}, \quad\quad\quad
	\epsu{d}{f} = \max_\eps \left\{dim_\eps(f) = d\right\}, \quad\quad\quad
	\dgap{d}{f} = \epsu{d}{f} - \epsu{d-1}{f},
	\end{align*}
	where we refer to $\dgap{d}{f}$ as the \em{approximate Fourier dimension gap} of dimension $d$. 
\end{definition}

The following theorem shows that (up to some slack in the dependence on the probability of error) the one-way communication complexity under the uniform distribution matches the linear sketch complexity. We note that the theorem can be applied to all possible values of $d$ and show how to pick specific values of $d$ of interest in Corollary~\ref{cor:linear-sketch-uniform}.
We illustrate tightness of Part 3 of this theorem in Appendix~\ref{app:tightness}.
We also note that the lower bounds given by this theorem are stronger than the basic extractor lower bound given in Appendix~\ref{sec:extractor}. See Remark~\ref{rem:extractor-vs-concentration} for further discussion.

\begin{theorem}\label{thm:linear-sketch-uniform}
	For any $f \colon \ftwo^n \to \oo$, $1 \le d \le n$ and $\epsilon_1 = \epsu{d}{f}$, $ \gamma < \frac{1 - \sqrt{\epsilon_1}}{2}$, $\delta = \dgap{d}{f} / 4$:
	\begin{align*}
1.\quad \distcu{(1 - \epsilon_1)/2}(\fplus{}) \le \distlu{(1 - \epsilon_1)/2}(f) \le d, \quad\quad\quad
2.\quad  \distlu{\gamma}(f) \ge d + 1, \quad\quad\quad
3.\quad  \distcu{\delta}(\fplus{}) \ge d.
	\end{align*}
\end{theorem}
\begin{proof}
\textbf{Part 1\footnote{This argument is a refinement of the standard ``sign trick'' from learning theory which approximates a Boolean function by taking a sign of its real-valued approximation under $\ell_2$.}.} 
By the assumptions of the theorem we know that there exists a $d$-dimensional subspace $A \le \ftwo^n$ which satisfies $\sum_{S \in A} \hat f^2(S) \ge \epsilon_1$.
Let $g \colon \ftwo^n \to \R$ be a function defined by its Fourier transform as follows:
\begin{align*}
\hat g(S) = 
\begin{cases}
\hat f(S) \text{, if } S \in A \\
0 \text{, otherwise}.
\end{cases}
\end{align*}
Consider drawing a random variable $\theta$ from the distribution with p.d.f $1 - |\theta|$ over $[-1,1]$.
\begin{proposition}\label{prop:theta-choice}
For all $t$ such that  $-1 \le t \le 1$ and $z \in \oo$ random variable $\theta$ satisfies:
$$\Pr_\theta[sgn(t - \theta) \neq z] \le \frac12 (z - t)^2.$$
\end{proposition}
\begin{proof}
W.l.o.g we can assume $z = 1$ as the case $z = -1$ is symmetric. Then we have:
$$\Pr_\theta[sgn(t - \theta) \neq 1] = \int_t^1 (1 - |\gamma|) d\gamma \le \int_t^1 (1 - \gamma) d\gamma= \frac12 (1 - t)^2.\qedhere$$
\end{proof}

Define a family of functions $g_\theta : \ftwo^n \to \oo$ as $g_\theta(x) = sgn(g(x) - \theta)$.
Then we have:
\begin{align*}
\E_\theta\left[\Pr_{x \sim \ftwo^n}[g_\theta(x) \neq f(x)]\right] &= \E_{x \sim \ftwo^n}\left[\Pr_{\theta}[g_\theta(x) \neq f(x)]\right]  \\
&= \E_{x \sim \ftwo^n}\left[\Pr_{\theta}[sgn(g(x) - \theta) \neq f(x)]\right] \\ 
&\le \E_{x \sim \ftwo^n}\left[\frac12 (f(x) - g(x))^2\right]  \text{(by Proposition~\ref{prop:theta-choice})}\\
&= \frac12 \|f - g\|_2^2.
\end{align*}
Using the definition of $g$ and Parseval we have:
\begin{align*}
\frac12 \|f - g\|_2^2  = \frac12 \|\widehat{f - g}\|_2^2  = \frac12 \|\hat{f} - \hat{g}\|_2^2 = \frac12 \sum_{S \notin A} \hat f^2(S)  \le \frac{1 - \epsilon_1}{2}.
\end{align*}
Thus, there exists a choice of $\theta$ such that $g_\theta$ achieves error at most $\frac{1 - \epsilon_1}{2}$.
Clearly $g_\theta$ can be computed based on the $d$ parities forming a basis for $A$ and hence $\distlu{(1 - \epsilon_1)/2}(f) \le d$.

\paragraph{Part 2.}
Fix any deterministic sketch that uses $d$ functions $\chi_{S_1}, \dots, \chi_{S_d}$ and let $S = (S_1, \dots, S_d)$.
For fixed values of these sketches $b = (b_1, \dots, b_d)$ where $b_i = \chi_{S_i}(x)$ we denote the restriction on the resulting coset as $f|_{(S,b)}$. Using the standard expression for the Fourier coefficients of an affine restriction the constant Fourier coefficient of the restricted function is given as:
\todo{promote this fact to prelims}
\begin{align*}
\widehat{f|_{(S,b)}}(\emptyset) = \sum_{Z \subseteq [d]} (-1)^{\sum_{i \in Z} b_i} \hat f\left(\sum_{i \in Z} S_i\right).
\end{align*}
Thus, we have:
\begin{align*}
\widehat{f|_{(S,b)}}(\emptyset)^2 =  \sum_{Z \subseteq [d]} \hat f^2(\sum_{i \in Z} S_i) + \sum_{Z_1 \neq Z_2 \subseteq [d]} (-1)^{\sum_{i \in Z_1 \Delta Z_2} b_i} \hat f(\sum_{i \in Z_1} S_i) \hat f(\sum_{i \in Z_2} S_i).
\end{align*}
Taking expectation over a uniformly random $b \sim U(\ftwo^d)$ we have:
\begin{align*}
\mathbb E_{b \sim U(\ftwo^d)}\left[\widehat{f|_{(S,b)}}(\emptyset)^2 \right] &= 
\mathbb E_{b \sim U(\ftwo^d)}\left[\sum_{Z \subseteq [d]} \hat f^2\left(\sum_{i \in Z} S_i\right)+ 
\sum_{Z_1 \neq Z_2 \subseteq [d]} (-1)^{\sum_{i \in Z_1 \Delta Z_2} b_i} \hat f\left(\sum_{i \in Z_1} S_i\right) \hat f\left(\sum_{i \in Z_2} S_i\right)\right] \\
&= \sum_{Z \subseteq [d]} \hat f^2\left(\sum_{i \in Z} S_i\right) .
\end{align*}

The latter sum is the sum of squared Fourier coefficients over a linear subspace of dimension $d$ and hence is at most $\epsilon_1$ by the assumption of the theorem.
Using Jensen's inequality: 
$$\mathbb E_{b \sim U(\ftwo^d)} \left[|\widehat{f|_{(S,b)}}(\emptyset)|\right]   \le \sqrt{\mathbb E_{b \sim U(\ftwo^d)}\left[\widehat{f|_{(S,b)}}(\emptyset)^2 \right]} \le \sqrt{\epsilon_1}.$$
For a fixed restriction $(S,b)$ if $|\hat f|_{(S,b)}(\emptyset)| \le \alpha$ then $|Pr[f|_{(S,b)} = 1] - Pr[f|_{(S,b)} = -1]| \le \alpha$ and hence no algorithm can predict the value of the restricted function on this coset with probability greater than $\frac{1 + \alpha}{2}$. 
Thus no algorithm can predict $f|_{(S_1, b_1), \dots, (S_d, b_d)}$ for a uniformly random choice of $(b_1, \dots, b_d)$ and hence also on a uniformly at random chosen $x$ with probability greater than $\frac{1 + \sqrt{\epsilon_1}}{2}$.

\paragraph{Part 3.} Let $\epsilon_2 = \epsu{d-1}{f}$ and recall that $\epsilon_1 = \epsu{d}{f}$.
\begin{definition}
	We say that $\mathcal A \le \ftwo^n$ \textit{distinguishes} $x_1, x_2 \in \ftwo^n$ if $\exists S \in \mathcal A \colon  \chi_{S}(x_1) \neq \chi_{S}(x_2).$
\end{definition}

We first prove the following auxiliary lemma.\todo{be careful with $\eps_1, \eps_2$.}

\begin{lemma}\label{lem:linear-sketch-concentration-lb}
	Fix $\epsilon_1 > \epsilon_2 \ge 0$ and $x_1, x_2 \in \ftwo^n$.
	If there exists a subspace $\mathcal A_d \le \ftwo^n$ of dimension $d$ which distinguishes $x_1$ and $x_2$ such that $f : \ftwo^n \to \oo$ is $\epsilon_1$-concentrated on $\mathcal A_d$ but is not $\epsilon_2$-concentrated on any $d-1$ dimensional linear subspace then: 
	$$\Pr_{z \in U(\ftwo^n)}[\fplus{x_1}(z) \neq \fplus{x_2}(z)] \ge \epsilon_1 - \epsilon_2.$$
\end{lemma}
\begin{proof}
	Note that for a fixed $x \in \ftwo^n$ (by Fact~\ref{fact:fourier-shift}) the Fourier expansion of $\fplus{x}$ can be given as:
	$$\fplus{x}(z) = \sum_{S \in \ftwo^n} \hat f(S) \chi_{S}(z + x) = \sum_{S \in \ftwo^n} \hat f(S) \chi_S(z) \chi_S(x).$$
	Thus we have:
	\begin{align*}
	\Pr_{z \in U(\ftwo^n)}[\fplus{x_1}(z) \neq \fplus{x_2}(z)] &= \frac{1}{2} \left(1 - \langle \fplus{x_1}, \fplus{x_2}\rangle\right)& \\
	& = \frac{1}{2} \left(1 - \left\langle \sum_{S_1 \in \ftwo^n} \hat f(S_1) \chi_{S_1} \chi_{S_1}(x_1), \sum_{S_2 \in \ftwo^n} \hat f(S_2) \chi_{S_2} \chi_{S_2}(x_2)\right\rangle\right)&\\
	& = \frac{1}{2} \left(1 -  \sum_{S \in \ftwo^n} \hat f(S)^2 \chi_{S}(x_1) \chi_{S}(x_2)\right) \text{(by orthogonality of characters)}& 
	\end{align*}
	
	We now analyze the expression $\sum_{S \in \ftwo^n} \hat f(S)^2 \chi_{S}(x_1) \chi_{S}(x_2)$.
	Breaking the sum into two parts we have:
	\begin{align*}
	\sum_{S \in \ftwo^n} \hat f(S)^2 \chi_{S}(x_1) \chi_{S}(x_2) &= \sum_{S \in \mathcal A_d} \hat f(S)^2 \chi_{S}(x_1) \chi_{S}(x_2) + \sum_{S \notin \mathcal A_d} \hat f(S)^2 \chi_{S}(x_1) \chi_{S}(x_2) \\
	& \le \sum_{S \in \mathcal A_d} \hat f(S)^2 \chi_{S}(x_1) \chi_{S}(x_2) + (1 - \epsilon_1).
	\end{align*}
	To give a bound on the first term we will use the fact that $\mathcal A_d$ distinguishes $x_1$ and $x_2$. We will need the following simple fact.
	\begin{proposition}\label{prop:good-basis}
		If $\mathcal A_d$ distinguishes $x_1$ and $x_2$ then there exists a basis $\mathcal S_1, \mathcal S_2, \dots, \mathcal S_d$ in $\mathcal A_d$ such that $\chi_{\mathcal S_1}(x_1) \neq \chi_{\mathcal S_1}(x_2)$ while $\chi_{\mathcal S_i}(x_1) = \chi_{\mathcal S_i}(x_2)$ for all $i \ge 2$.
	\end{proposition}
	\begin{proof}
		Since $\mathcal A_d$ distinguishes $x_1$ and $x_2$ there exists a $\mathcal S \in \mathcal A_d$ such that $\chi_{\mathcal S}(x_1) \neq \chi_{\mathcal S}(x_2)$.
		Fix $\mathcal S_1 = \mathcal S$ and consider an arbitrary basis in $\mathcal A_d$ of the form $(\mathcal S_1, \mathcal T_2, \dots, \mathcal T_d)$.
		For $i \ge 2$ if $\chi_{\mathcal T_i}(x_1) = \chi_{\mathcal T_i}(x_2)$ then we let $\mathcal S_i = \mathcal T_i$. Otherwise, we let $\mathcal S_i = \mathcal T_i + \mathcal S_1$, which preserves the basis and ensures that: 
		$$
		\chi_{\mathcal S_i}(x_1) = \chi_{\mathcal T_i + \mathcal S_1}(x_1) = \chi_{\mathcal T_i}(x_1) \chi_{\mathcal S_1}(x_1) = \chi_{\mathcal T_i}(x_1) \chi_{\mathcal S_1}(x_2) = \chi_{\mathcal S_i}(x_2).\qedhere
		$$
	\end{proof}
	
	Fix the basis $(\mathcal S_1, \mathcal S_2, \dots, \mathcal S_d)$ in $\mathcal A_d$ with the properties given by Proposition~\ref{prop:good-basis}.
	Let $\mathcal A_{d-1} = span(\mathcal S_2, \dots, \mathcal S_d)$ so that for all $S \in \mathcal A_{d-1}$ it holds that $\chi_S(x_1) = \chi_S(x_2)$.
	Then we have:
	\begin{align*}
	\sum_{S \in \mathcal A_d} \hat f(S)^2 \chi_{S}(x_1) \chi_{S}(x_2) & = \sum_{S \in \mathcal A_{d-1}} \hat f(S)^2 \chi_S(x_1) \chi_S(x_2) + \sum_{S \in \mathcal A_{d-1}} \hat f(S + \mathcal S_1)^2 \chi_{S + \mathcal S_1}(x_1) \chi_{S + \mathcal S_1}(x_2)\\
	&= \sum_{S \in \mathcal A_{d-1}} \hat f(S)^2  - \sum_{S \in \mathcal A_{d-1}} \hat f(S + \mathcal S_1)^2 \\ 
	\end{align*}
	The first term in the above summation is at most $\epsilon_2$ since $f$ is not $\epsilon_2$-concentrated on any $(d-1)$-dimensional linear subspace.
	The second is at least $\epsilon_1 - \epsilon_2$ since $f$ is $\epsilon_1$-concentrated on $\mathcal A_d$.
	
	Thus, putting things together we have that 
	\begin{align*}
	\sum_{S \in \ftwo^n} \hat f(S)^2 \chi_{S}(x_1) \chi_{S}(x_2) \le \epsilon_2 - (\epsilon_1 - \epsilon_2) + (1 - \epsilon_1) = 1 - 2 (\epsilon_1 - \epsilon_2). 
	\end{align*}
	This completes that proof showing that $\Pr_{z \in U(\ftwo^n)}[\fplus{x_1}(z) \neq \fplus{x_2}(z)] \ge \epsilon_1 - \epsilon_2$.
\end{proof}

We are now ready to complete the proof of the third part of Theorem~\ref{thm:linear-sketch-uniform}.
	We can always assume that the protocol that Alice uses is deterministic since for randomized protocols one can fix their randomness to obtain the deterministic protocol with the smallest error.
	Fix a $(d - 1)$-bit deterministic protocol that Alice is using to send a message to Bob.
	This protocol partitions the rows of the communication matrix into $t = 2^{d-1}$ rectangles corresponding to different messages. \todo{make this more formal and sync with CC notation}
	 We denote the sizes of these rectangles as $r_1, \dots, r_t$ and the rectangles themselves as $R_1, \dots, R_t \subseteq \ftwo^n$ respectively.
	Let the outcome of the protocol be $P(x,y)$. Then the error is given as:
	\begin{align*}
	\E_{x,y \sim U(\ftwo^n)}\left[\mathbb{1} [P(x,y) \neq f(x + y)]\right] &= \sum_{i = 1}^t \frac{r_i}{2^n} \times \E_{x \sim U(R_i) ,y \sim U(\ftwo^n)}\left[\mathbb{1} [P(x,y) \neq f(x + y)]\right]
	\\
	&\ge  \sum_{i \colon r_i > 2^{n-d}} \frac{r_i}{2^n} \times \E_{x \sim U(R_i) ,y \sim U(\ftwo^n)}\left[\mathbb{1} [P(x,y) \neq f(x + y)]\right],
	\end{align*}
where we only restricted attention to rectangles of size greater than $2^{n - d}$.
Our next lemma shows that in such rectangles the protocol makes a significant error:
	
	\begin{lemma}\label{lem:rectangle-error}
		If $r_i > 2^{n-d}$ then:
		 $$\E_{x \sim U(R_i) ,y \sim U(\ftwo^n)}\left[\mathbb{1} [P(x,y) \neq f(x + y)]\right] \ge \frac12 \frac{r_i - 2^{n - d}}{r_i} (\epsilon_1 - \epsilon_2).$$
	\end{lemma}
	\begin{proof}
		For $y \in \ftwo^n$ let  $p_y(R_i) = \min(\Pr_{x \sim U(R_i)}[f(x + y) = 1], \Pr_{x \sim U(R_i)}[f(x + y) = -1])$.
		We have:
		\begin{align*}
		\E_{x \sim U(R_i) ,y \sim U(\ftwo^n)}\left[\mathbb{1} [P(x,y) \neq f(x + y)]\right] &=  \E_{y \sim U(\ftwo^n)} \E_{x \sim U(R_i)}\left[\mathbb{1} [P(x,y) \neq f(x + y)]\right]\\ 
		&\ge \E_{y \sim U(\ftwo^n)} \left[p_y(R_i)\right] \\
		&\ge \E_{y \sim U(\ftwo^n)} \left[p_y(R_i) (1 - p_y(R_i))\right] \\
		&= \E_{y \sim U(\ftwo^n)} \left[\frac12 \Pr_{x_1, x_2 \sim U(R_i)}\left[f(x_1 + y) \neq f(x_2 + y)\right] \right]\\
		&= \frac12 \E_{x_1, x_2 \sim U(R_i)} \left[ \E_{y \sim U(\ftwo^n)}\left[\mathbb 1 \left[f(x_1 + y) \neq f(x_2 + y)\right]\right] \right]
		\end{align*}
		
		Fix a $d$-dimensional linear subspace $\mathcal A_d$ such that $g$ is $\epsilon_1$-concentrated on $\mathcal A_d$. 
		There are $2^{n-d}$ vectors which have the same inner products with all vectors in $\mathcal A_d$.
		Thus with probability at least $\frac{r_i - 2^{n-d}}{r_i}$ two random vectors $x_1, x_2 \sim U(R_i)$ are distinguished by $\mathcal A_d$.
		Conditioning on this event we have:
		\begin{align*}
		&\frac12 \E_{x_1, x_2 \sim U(R_i)} \left[ \E_{y \sim U(\ftwo^n)}\left[\mathbb 1 \left[f(x_1 + y) \neq f(x_2 + y)\right]\right] \right] \\
		&\ge \frac12  \frac{r_i - 2^{n - d}}{r_i} \E_{y \sim U(\ftwo^n)}\left[\mathbb 1 \left[f(x_1 + y) \neq f(x_2 + y)\right] | \text{$\mathcal A_d$ distinguishes $x_1, x_2$} \right] \\
		& \ge \frac12 \frac{r_i - 2^{n - d}}{r_i} (\epsilon_1 - \epsilon_2),
		\end{align*}
		where the last inequality follows by Lemma~\ref{lem:linear-sketch-concentration-lb}.
	\end{proof}
	
	Using Lemma~\ref{lem:rectangle-error} we have:
	\begin{align*}
	\E_{x,y \sim U(\ftwo^n)}\left[\mathbb{1} [P(x,y) \neq f(x + y)]\right] &\ge \frac{\epsilon_1 - \epsilon_2}{2^{n + 1}} \sum_{i \colon r_i > 2^{n - d}}   \left(r_i - 2^{n - d}\right) \\
	&= \frac{\epsilon_1 - \epsilon_2}{2^{n + 1}} \left(\sum_{i =1}^t   \left(r_i - 2^{n - d}\right) - \sum_{i \colon r_i \le 2^{n - d}}   \left(r_i - 2^{n - d}\right)\right) \\
	& \ge \frac{\epsilon_1 - \epsilon_2}{2^{n + 1}} \left(2^n - 2^{n-1}\right) \\
	& = \frac{\epsilon_1 - \epsilon_2}{4},
	\end{align*}
	where the inequality follows since $\sum_{i = 1}^t r_i = 2^n$, $t = 2^{d-1}$ and all the terms in the second sum are non-positive.
\end{proof}

An important question that arises when applying Theorem~\ref{thm:linear-sketch-uniform} is the choice of the value of $d$.
The following simple corollaries of Theorem~\ref{thm:linear-sketch-uniform} give one particularly simple way of choosing these values for any function in such a way that we obtain a non-trivial lower bound for $O(1/n)$-error.

\begin{corollary}\label{cor:linear-sketch-uniform}
	For any $f \colon \ftwo^n \to \oo$ such that $\hat f(\emptyset) \le \theta$ for some constant $\theta < 1$ there exists an integer $d \ge 1$ such that:
	 \begin{align*}
	 \distcu{\Theta(\frac1n)}(\fplus{}) \ge d \ge \distlu{\frac13}(f)
	 \end{align*}
\end{corollary}
\begin{proof}
	%Let $A_d$ be a family of all linear subspaces of dimension $d$ in $\ftwo^n$.
	We have $\epsu{0}{f} < \theta$ and $\epsu{n}{f} = 1$.
	Let $d^*=\argmax_{d = 1}^n \dgap{d}{f}$ and $\Delta(f) = \dgap{d^*}{f}$.
	Consider cases:
	
	\textbf{Case 1.} $\Delta(f) \ge \frac{1 - \theta}{3}$. By Part 3 of Theorem~\ref{thm:linear-sketch-uniform} we have that $\distcu{\frac{1 - \theta}{12n}}(\fplus{}) \ge d^*$.
	Furthermore, $\epsu{d^*}{f} \ge \theta +\epsu{d^*}{f} - \epsu{d^*-1}{f} = \theta + \Delta(f) \ge \frac13 - \frac{2\theta}{3}$.
	By Part 1 of Theorem~\ref{thm:linear-sketch-uniform} we have $\distlu{\frac{1 - \theta}{3}}(f) \le d^*$.
	
	\textbf{Case 2.} $\Delta(f) < \frac{1 - \theta}{3}$.
		In this case there exists $d_1 \ge 1$ such that $\epsu{d_1}{f} \in [\theta_1, \theta_2]$ where $\theta_1 = \theta + \frac{1 - \theta}{3}, \theta_2 = \theta + \frac{2(1 - \theta)}{3}$.
		By averaging there exists $d_2 > d_1$ such that $\dgap{d_2}{f} = \epsu{d_2}{f} - \epsu{d_2 - 1}{f}\ge \frac{1 - \theta_2}{n}  = \Theta(\frac1n)$.
	Applying Part 3 of Theorem~\ref{thm:linear-sketch-uniform} we have that $\distcu{\Theta(\frac1n)}(\fplus{}) \ge d_2$.
	Furthermore, we have $\epsu{d_2}{f} \ge \theta_1$ and hence $\frac{1 - \epsu{d_2}{f}}{2} \le \frac{1 - \theta_1}{2} < \frac{1 - \theta}{3}$. By Part 1 of Theorem~\ref{thm:linear-sketch-uniform} we have $\distlu{\frac{1 - \theta}{3}}(f) \le d_2$.
\end{proof}

The proof of Theorem~\ref{thm:linear-sketch-uniform-main} follows directly from Corollary~\ref{cor:linear-sketch-uniform}.
If $\theta \le \frac13$ then the statement of the theorem holds.
If $\theta \ge \frac13$ then $\epsu{0}{f} \ge \frac13$ so by Part 1 of Theorem~\ref{thm:linear-sketch-uniform} we have $\distlu{\frac13}(f) \le 0$ and the inequality holds trivially.

Furthermore, using the same averaging argument as in the proof of Corollary~\ref{cor:linear-sketch-uniform} we obtain the following generalization of the above corollary that will be useful for our applications.
\begin{corollary}\label{cor:linear-sketch-uniform-lb}
	For any $f \colon \ftwo^n \to \oo$ and $d$ such that $\epsu{d - 1}{f} \le \theta$ it holds that:  
	$$\distcu{\frac{1 - \theta}{4 (n - d)}}(f) \ge d.$$
\end{corollary}

%\todo{Add applications for streaming and distributed under unifrom distribution?}

\section{Applications}\label{sec:applications}
\subsection{Composition theorem for majority}\label{sec:rec-majority}
In this section using Theorem~\ref{thm:linear-sketch-uniform} we give a composition theorem for $\ftwo$-sketching of the composed $Maj_3$ function.
Unlike in the deterministic case for which the composition theorem is easy to show (see Lemma~\ref{lem:det-sketch-composition}) in the randomized case composition results require more work.
\begin{definition}[Composition]
	For $f \colon \ftwo^n \to \ftwo$ and $g \colon \ftwo^m \rightarrow \ftwo$ their composition $f \circ g \colon \ftwo^{mn} \rightarrow \ftwo$ is defined as:
	\begin{align*}
	(f \circ g)(x) = f(g(x_1, \dots, x_m), g(x_{m + 1}, \dots, x_{2m}), \dots, g(x_{m(n - 1) + 1}, \dots, x_{mn})).
	\end{align*}
\end{definition}

Consider the recursive majority function $\mtk{k} \equiv \mt \circ \mt \circ \dots \circ \mt$ where the composition is taken $k$ times.

\begin{theorem}\label{thm:rand-sketch-rec-majority}
For any $d\le n$ and $k = \log_3 n$ it holds that $\epsu{d}{\mtk{k}} \le \frac{4d}{n}$.
\end{theorem}

%\begin{theorem}\label{thm:rand-sketch-rec-majority}
%	No randomized linear sketch of size less than $\frac{\epsilon^2}{n}$ allows to predict for $\mtk{k}$ with probability better than $\frac{1 - \epsilon}{2}$.
%\end{theorem}

	First, we show a slighthly stronger result for standard subspaces and then extend this result to arbitrary subspaces with a loss of a constant factor.
	Fix any set $S \subseteq [n]$ of variables.
	We associate this set with a collection of standard unit vectors corresponding to these variables.
	Hence in this notation $\emptyset$ corresponds to the all-zero vector.
	
		\begin{lemma}\label{lem:recursive-majority-standard-subspace}
		For any standard subspace whose basis consists of singletons from the set $S \subseteq[n]$ it holds that:
		\begin{align*}
		\sum_{Z \in span(S)} \left(\widehat{\mtk{k}} (Z)\right)^2 \le \frac{|S|}{n}
		\end{align*}
	\end{lemma}
	\begin{proof}
		The Fourier expansion of $\mt$ is given as $\mt(x_1, x_2, x_3) = \frac{1}{2}\left(x_1 + x_2 + x_3 - x_1 x_2 x_3\right)$.
		 For $i \in \{1,2,3\}$ let $N_i = \{(i - 1) n/3 + 1, \dots, i n/3\}$. Let $S_i = S \cap N_i$.
		Let $\alpha_i$ be defined as:
		\begin{align*}
		\alpha_i = \sum_{Z \in span(S_i)} \left(\widehat{\mt^{\circ {k - 1}}} (Z)\right)^2.
		\end{align*}
		Then we have:
		\begin{align*}
		\sum_{Z \in span(S)} \left(\widehat{\mtk{k}} (Z)\right)^2 =  \sum_{i = 1}^3 \sum_{Z \in span(S_i)} \left(\widehat{\mtk{k}} (Z)\right)^2 + \sum_{Z \in span(S) - \cup_{i = 1}^3 span(S_i)} \left(\widehat{\mtk{k}} (Z)\right)^2.
		\end{align*} 
	 For each $S_i$  we have 
	 $$ \sum_{Z \in span(S_i)} \left(\widehat{\mtk{k}} (Z)\right)^2 = \frac{1}{4}\sum_{Z \in span(S_i)} \left(\widehat{\mt^{\circ k - 1}} (Z)\right)^2 = \frac{\alpha_i}{4}.$$
		Moreover,  for each $Z \in span(S) - \cup_{i = 1}^3 span(S_i)$ we have: 
		\begin{align*}
		\widehat{\mtk{k}}(Z) = \begin{cases}
		- \frac{1}{2} \widehat{Maj_{3}^{\circ k -1}}(Z_1)\widehat{Maj_{3}^{\circ k -1}}(Z_2)\widehat{Maj_{3}^{\circ k -1}}(Z_3) & \text{if } Z \in \times_{i = 1}^3 (span(S_i) \setminus \emptyset)\\
		0 & \text{otherwise.}
		\end{cases}
		\end{align*}
		Thus, we have:
		\begin{align*}
		&\sum_{Z \in (span(S_1) \setminus \emptyset) \times (span(S_2) \setminus \emptyset) \times (span(S_3) \setminus \emptyset)} \left(\widehat{\mtk{k}}(Z) \right)^2  \\
		&=\sum_{Z \in (span(S_1) \setminus \emptyset) \times (span(S_2) \setminus \emptyset) \times (span(S_3) \setminus \emptyset)}\frac{1}{4} \left(\widehat{Maj_{3}^{\circ k -1}}(Z_1)\right)^2 \left(\widehat{Maj_{3}^{\circ k -1}}(Z_2)\right)^2\left(\widehat{Maj_{3}^{\circ k -1}}(Z_3)\right)^2 \\
		&=\frac{1}{4}    \sum_{Z \in (span(S_1) \setminus \emptyset) } \left(\widehat{Maj_{3}^{\circ k -1}}(Z_1)\right)^2 \sum_{Z \in (span(S_2) \setminus \emptyset) } \left(\widehat{Maj_{3}^{\circ k -1}}(Z_2)\right)^2\sum_{Z \in (span(S_3) \setminus \emptyset) } \left(\widehat{Maj_{3}^{\circ k -1}}(Z_3)\right)^2 \\
		&= \frac{1}{4} \alpha_1 \alpha_2 \alpha_3.
		\end{align*}
		where the last equality holds since $\widehat{Maj_{3}^{\circ k -1}}(\emptyset) = 0$.
		Putting this together we have:
		\begin{align*}
		&\sum_{Z \in span(S)} \left(\widehat{\mtk{k}} (Z)\right)^2 = \frac{1}{4}(\alpha_1 + \alpha_2 + \alpha_3 + \alpha_1 \alpha_2 \alpha_3) \\
		&\le \frac{1}{4}\left(\alpha_1 + \alpha_2 + \alpha_3 + \frac{1}{3}(\alpha_1 + \alpha_2 + \alpha_3)\right) = \frac{1}{3}(\alpha_1 + \alpha_2 + \alpha_3).
		\end{align*}
		Applying this argument recursively to each $\alpha_i$ for $k-1$ times we have:
		\begin{align*}
		\sum_{Z \in span(S)} \left(\widehat{\mtk{k}} (Z)\right)^2 \le \frac{1}{3^k} \sum_{i = 1}^{3^k} \gamma_i,
		\end{align*}
		where $\gamma_i = 1$ if $i \in S$ and $0$ otherwise. Thus, $\sum_{Z \in span(S)} \left(\widehat{\mtk{k}} (Z)\right)^2 \le \frac{|S|}{n}$.
	\end{proof}
	
%	By Lemma~\ref{lem:recursive-majority-standard-subspace} and Lemma~\ref{lem:rand-sketch-complexity} no subset of size less than $\epsilon^2 n$ can be used to predict $\mtk{k}$ with probability better than $\frac{1 - \epsilon}{2}$.

	To extend the argument to arbitrary linear subspaces we show that any such subspace has less Fourier weight than a collection of three carefully chosen standard subspaces.
	First we show how to construct such subspaces in Lemma~\ref{lem:standard-subspace-domination}.
		
	For a linear subspace $L \le \ftwo^n$ we denote the set of all vectors in $L$ of odd Hamming weight as $\mathcal O(L)$ and refer to it as the \textit{odd set} of $L$.
	For two vectors $v_1, v_2 \in \ftwo^n$ we say that $v_1$ \textit{dominates} $v_2$ if the set of non-zero coordinates of $v_1$ is a (not necessarily proper) subset of the set of non-zero coordinates of $v_2$.
	For two sets of vectors $S_1, S_2 \subseteq \ftwo^n$ we say that $S_1$ \textit{dominates} $S_2$ (denoted as $S_1 \prec S_2$) if there is a matching $M$ between $S_1$ and $S_2$ of size $|S_2|$ such that for each $(v_1 \in S_1, v_2 \in S_2) \in M$ the vector $v_1$ dominates $v_2$.
	
	\begin{lemma}[Standard subspace domination lemma]\label{lem:standard-subspace-domination}
		For any linear subspace $L \le \ftwo^n$ of dimension $d$ there exist three standard linear subspaces $S_1, S_2, S_3 \le \ftwo^n$ such that:
		$$\mathcal O(L) \prec \mathcal O(S_1) \cup \mathcal O(S_2) \cup \mathcal O(S_3),$$
		and $dim(S_1) = d-1$, $dim(S_2) = d$, $dim(S_3) = 2d$.
	\end{lemma}
	\begin{proof}
		Let $A \in \ftwo^{d \times n}$ be the matrix with rows corresponding to the basis in $L$.
		We will assume that $A$ is normalized in a way described below.
		First, we apply Gaussian elimination to ensure that $A = (I, M)$ where $I$ is a $d \times d$ identity matrix.
		If all rows of $A$ have even Hamming weight then the lemma holds trivially since $\mathcal O(L) = \emptyset$.
		By reordering rows and columns of $A$ we can always assume that for some $k \ge 1$ the first $k$ rows of $A$ have odd Hamming weight and the last $d - k$ have even Hamming weight.
		Finally, we add the first column to each of the last $d - k$ rows, which makes all rows have odd Hamming weight.
		This results in $A$ of the following form:
		\[
		A = \left(
		\begin{array}{c|c|c|c}
		1 & 0 \cdots 0 & 0 \cdots 0 & a\\ \hline
		0 & \raisebox{-15pt}{{\huge\mbox{{$I_{k-1}$}}}} & \raisebox{-15pt}{{\huge\mbox{{$0$}}}} & \raisebox{-15pt}{{\huge\mbox{{$M_1$}}}} \\[-4ex]
		\vdots & & &  \\[-0.5ex]
		0 & & &  \\ \hline 
		1 & \raisebox{-15pt}{{\huge\mbox{{$0$}}}} & \raisebox{-15pt}{{\huge\mbox{{$I_{d - k}$}}}} & \raisebox{-15pt}{{\huge\mbox{{$M_2$}}}} \\[-4ex]
		\vdots & & & \\[-0.5ex]
		1 & & &
		\end{array}
		\right)
		\]
		We use the following notation for submatrices: $A[i_1,j_1;i_2,j_2]$ refers to the submatrix of $A$ with rows between $i_1$ and $j_1$ and columns between $i_2$ and $j_2$ inclusive.
		We denote to the first row as $v$,  the submatrix $A[2,k;1,n]$ as $\mathcal A$ and the submatrix $A[k +1, d; 1,n]$ as $\mathcal B$. Each $x \in \mathcal O(L)$ can be represented as $\sum_{i \in S} A_i$ where the set $S$ is of odd size and the sum is over $\ftwo^n$.
		We consider the following three cases corresponding to different types of the set $S$.
		
		\textbf{Case 1.} \todo{might be possible to simplify these cases} $S \subseteq rows(\mathcal A) \cup rows(\mathcal B)$. This corresponds to all odd size linear combinations of the rows of $A$ that don't include the first row.
		Clearly, the set of such vectors is dominated by $\mathcal O(S_1)$ where $S_1$ is the standard subspace corresponding to the span of the rows of the submatrix $A[2,d;2,d]$.
		
		\textbf{Case 2.} $S$ contains the first row, $|S \cap rows(\mathcal A)|$ and $|S \cap rows(\mathcal B)|$ are even.
		All such linear combinations have their first coordinate equal $1$.
		Hence, they are dominated by a standard subspace corresponding to span of the rows the $d \times d$ identity matrix,  which we refer to as $S_2$.
		
		\textbf{Case 3.} $S$ contains the first row, $|S \cap rows(\mathcal A)|$ and $|S \cap rows(\mathcal B)|$ are odd.
		All such linear combinations have their first coordinate equal $0$.
		This implies that the Hamming weight of the first $d$ coordinates of such linear combinations is even and hence the other coordinates can't be all equal to $0$.
		Consider the submatrix $M = A[1,d; d + 1,n]$ corresponding to the last $n - d$ columns of $A$.
		Since the rank of this matrix is at most $d$ by running Gaussian elimination on $M$ we can construct a matrix $M'$ containing as rows the basis for the row space of $M$ of the following form:
		\begin{align*}
		M' = 
		\begin{pmatrix}
		I_t & M_1 \\
		0 &  0
		\end{pmatrix}
		\end{align*}
		where $t = rank(M)$.
		This implies that any non-trivial linear combination of the rows of $M$ contains $1$ in one of the first $t$ coordinates.
		We can reorder the columns of $A$ in such a way that these $t$ coordinates have indices from $d + 1$ to $d + t$.
		Note that now the set of vectors spanned by the rows of the $(d + t) \times (d + t)$ identity matrix $I_{d + t}$ dominates the set of linear combinations we are interested in.
		Indeed, each such linear combination has even Hamming weight in the first $d$ coordinates and has at least one coordinate equal to $1$ in the set $\{d +1, \dots, d + t\}$. This gives a vector of odd Hamming weight that dominates such linear combination.
		Since this mapping is injective we have a matching. 
		We denote the standard linear subspace constructed this way as $S_3$ and clearly $dim(S_3) \le 2d$.
	\end{proof}

%	To extend our argument to arbitrary linear sketches consider such sketch in a standard form $(I_d, M)$ where $d$ is the dimension of the sketch.
%	Let $S$ be the set of variables corresponding to the columns of $I_d$ in the sketch.
%	Define $S_i = N_i \cap S$ as in the proof of Lemma~\ref{lem:recursive-majority-standard-subspace}.
%	Note that all linear combinations of the rows of a linear sketch can be bijectively mapped to the linear combinations of the rows restricted to the first $d$ columns.
%	The key property of this mapping is that it maps vectors to vectors whose set of non-zero coordinates is a (not necessarily strict) subset of the set of non-zero coordinates of the original vector.
%	Thus, if we can show that such mapping only increases the sum of squared Fourier coefficients in the subspace spanned by the linear sketch on a term by term basis then the general case is reduce to the case already handled by Lemma~\ref{lem:recursive-majority-standard-subspace}.

	The following proposition shows that the spectrum of the $\mtk{k}$ is monotone decreasing under inclusion if restricted to odd size sets only:
	\begin{proposition}\label{prop:fourier-monotonicity-rec-majority}
		For any two sets $Z_1 \subseteq Z_2$ of odd size it holds that:
		\begin{align*}
		\left|\widehat{\mtk{k}}(Z_1)\right| \ge \left|\widehat{\mtk{k}}(Z_2)\right|.
		\end{align*}
	\end{proposition}
	\begin{proof}
		The proof is by induction on $k$. 
		Consider the Fourier expansion of $\mt(x_1, x_2, x_3) = \frac{1}{2}(x_1 + x_2 + x_3 - x_1 x_2 x_3)$. The case $k=1$ holds since all Fourier coefficients have absolute value $1/2$.
		Since $\mtk{k} = \mt \circ (\mtk{k-1})$ all Fourier coefficients of $\mtk{k}$ result from substituting either a linear or a cubic term in the Fourier expansion by the multilinear expansions of $\mtk{k-1}$.
		This leads to four cases.
		
		\textbf{Case 1.} $Z_1$ and $Z_2$ both arise from linear terms.
		In this case if $Z_1$ and $Z_2$ aren't disjoint then they arise from the same linear term and thus satisfy the statement by the inductive hypothesis.  
		
		\textbf{Case 2.} If $Z_1$ arises from a cubic term and $Z_2$ from the linear term then it can't be the case that $Z_1 \subseteq Z_2$ since $Z_2$ contains some variables not present in $Z_1$.
		
		\textbf{Case 3.} If $Z_1$ and $Z_2$ both arise from the cubic term then we have $(Z_1 \cap N_i) \subseteq (Z_2 \cap N_i)$ for each $i$.
		By the inductive hypothesis we then have $\left|\widehat{\mtk{k-1}}(Z_1 \cap N_i)\right| \ge \left|\widehat{\mtk{k-1}}(Z_2 \cap N_i)\right|$.
		Since for $j = 1,2$ we have $\widehat{\mtk{k}}(Z_j) = - \frac{1}{2} \prod_i \widehat{\mtk{k-1}}(Z_j \cap N_i)$ the desired inequality follows.
		
		\textbf{Case 4.} If $Z_1$ arises from the linear term and $Z_2$ from the cubic term then w.l.o.g. assume that $Z_1$ arises from the $x_1$ term.
		Note that $Z_1 \subseteq (Z_2 \cap N_1)$ since $Z_1 \cap (N_2 \cup N_3) = \emptyset$.
		By the inductive hypothesis applied to $Z_1$ and $Z_2 \cap N_1$ the desired inequality holds.

	\end{proof}

We can now complete the proof of Theorem~\ref{thm:rand-sketch-rec-majority}

\begin{proof}[Proof of Theorem~\ref{thm:rand-sketch-rec-majority}]
By combining Proposition~\ref{prop:fourier-monotonicity-rec-majority}
and Lemma~\ref{lem:recursive-majority-standard-subspace} we have that any set $\mathcal T$ of vectors that is dominated by $\mathcal O(\mathcal S)$ for some standard subspace $\mathcal S$ satisfies $\sum_{S \in \mathcal T} \widehat{\mtk{k}}(S)^2 \le \frac{dim(\mathcal S)}{n}$.
By the standard subspace domination lemma (Lemma~\ref{lem:standard-subspace-domination}) any subspace $L \le \ftwo^n$ of dimension $d$ has $\mathcal O(L)$ dominated by a union of three standard subspaces of dimension $2d$, $d$ and $d - 1$ respectively.
Thus, we have $\sum_{S \in \mathcal O(L)} \widehat{\mtk{k}}(S)^2 \le\frac{2d}{n} +\frac{d}{n} + \frac{d - 1}{n} \le  \frac{4d}{n}$.
\end{proof}

We have the following corollary of Theorem~\ref{thm:rand-sketch-rec-majority} that proves Theorem~\ref{thm:rec-majority}.
\begin{corollary}
For any $\epsilon \in [0,1]$, $\gamma < \frac12 - \epsilon$ and $k = \log_3 n$ it holds that:
\begin{align*}
\distlu{\gamma}(\mtk{k}) \ge \epsilon^2 n + 1, \quad\quad\quad \distcu{\frac1n \left(\frac14 - \eps^2\right)}({\mtk{k}}^+) \ge \epsilon^2 n + 1. 
\end{align*} 
\end{corollary}
\begin{proof}
Fix $d = \epsilon^2 n$. For this choice of $d$ Theorem~\ref{thm:rand-sketch-rec-majority} implies that $\epsilon_d(\mtk{k}) \le 4\epsilon^2$.
The first part follows from Part 2 of Theorem~\ref{thm:linear-sketch-uniform}.
The second part is by Corollary~\ref{cor:linear-sketch-uniform-lb} as by taking $\eps = \sqrt{d/n}$ we can set $\theta  =  4 \eps^2 \ge \epsu{d}{\mtk{k}}$ and hence:
 \begin{align*}
\eps^2n  +1 \le \distcu{\frac{1 - \theta}{4(n - d)}}(\mtk{k}) =  \distcu{\frac{1 - 4 \epsilon^2}{4 n (1 - \epsilon^2)}}(\mtk{k}) \le \distcu{\frac1n \left(\frac14 - \eps^2\right)}(\mtk{k}).
 \end{align*}

\end{proof}

\subsection{Address function and Fourier sparsity}\label{sec:address}
Consider the \textit{addressing function} $Add_n \colon \{0,1\}^{\log n + n} \rightarrow \{0,1\}$\footnote{In this section it will be more convenient to represent both domain and range of the function using $\{0,1\}$ rather than $\ftwo$.} defined as follows:
\begin{align*}
Add_n(x, y_1, \dots, y_n) = y_x, \text{ where } x\in\{0,1\}^{\log n}, y_i \in \{0,1\},
\end{align*}
i.e. the value of $Add_n$ on an input $(x,y)$ is given by the $x$-th bit of the vector $y$ where $x$ is treated as a binary representation of an integer number in between $1$ and $n$.
Addressing function has only $n^2$ non-zero Fourier coefficients. In fact, as shown by Sanyal~\cite{S15} Fourier dimension, and hence by Fact~\ref{fact:fourier-dimension} also the deterministic sketch complexity, of any Boolean function with Fourier sparsity $s$ is $O(\sqrt{s} \log s)$.

Below using the addressing function we show that this relationship is tight (up to a logarithmic factor) even if randomization is allowed, i.e. even for a function with Fourier sparsity $s$ an $\ftwo$ sketch of size $\Omega(\sqrt{s})$ might be required.

\begin{theorem}\label{thm:rand-sketch-complexity-address}
	For the addressing function $Add_n$ and values $1 \le d \le n$ and $\epsilon =d/n$ it holds that:
	$$\distlu{\frac{1 - \sqrt{\epsilon}}{2}}(Add_n^+) \ge d, \quad\quad\quad \distcu{\Theta(\frac{1-\epsilon}{n})}(Add_n) \ge d.$$
\end{theorem}
\begin{proof}
	If we apply the standard Fourier notaion switch where we replace $0$ with $1$ and $1$ with $-1$ in the domain and the range of the function then the addressing function $Add_n(x,y)$ can be expressed as the following multilinear polynomial:
	\begin{align*}
	Add_n(x,y) = \sum_{i \in \{0,1\}^{\log n}} y_i \prod_{j \colon i_j =1} \left(\frac{1 - x_j}{2}\right) \prod_{j \colon i_j =0} \left(\frac{1 + x_j}{2}\right),
	\end{align*}
	which makes it clear that the only non-zero Fourier coefficents correspond to the sets that contain a single variable from the addressee block and an arbitrary subset of variables from the address block.
	This expansion also shows that the absolute value of each Fourier coefficient is equal to $\frac{1}{n}$.
	
	Fix any $d$-dimensional subspace $\mathcal A_d$ and consider the matrix $M \in \ftwo^{d \times {(\log n +n)}}$ composed of the basis vectors as rows.
	We add to $M$ extra $\log n$ rows which contain an identity matrix in the first $\log n$ coordinates and zeros everywhere else.
	This gives us a new matrix $M' \in \ftwo^{(d + \log n) \times {(\log n +n)}}$. Applying Gaussian elimination to $M'$ we can assume that it is of the following form:
			\begin{align*}
			M' = 
			\begin{pmatrix}
			I_{\log n} & 0&0 \\
			0 & I_{d'} & M \\
			0&0&0
			\end{pmatrix},
			\end{align*}
	where $d' \le d$. Thus, the total number of non-zero Fourier coefficients spanned by the rows of $M'$ equals $n d'$.
	Hence, the total sum of squared Fourier coeffients in $\mathcal A_d$ is at most $\frac{d'}{n} \le \frac{d}{n}$, i.e. $\epsu{d}{Add_n} \le \frac{d}{n}$.
		By Part 2 of Theorem \ref{thm:linear-sketch-uniform} and Corollary~\ref{cor:linear-sketch-uniform-lb} the statement of the theorem follows.

%	W.l.o.g we can assume that this matrix is diagonalized and is in the standard form $(I_d, M')$ where $I_d$ is a $d \times d$ identity matrix and $M'$ is a $d \times (n - d)$-matrix.
%	Clearly, any linear combination of more than two rows of $M$ has Hamming weight greater than two just from the contribution of the first $d$ coordinates.
%	Thus, we have $|W_2 \cap \mathcal A_d| \le \binom{d}{2}$ and the overall Fourier weight in the subspace is at most $\frac{\binom{d}{2}}{n^2} \le \frac{d^2}{2 n^2}$.
%	Hence, $Add_n$ isn't $\epsilon$-concentrated on any affine subspace of dimension less than $\sqrt{2 \epsilon} n$ and the theorem follows from Lemma~\ref{lem:rand-sketch-complexity}.
\end{proof}

\subsection{Symmetric functions}\label{sec:symmetric}
A function $f:\ftwo^n \to \ftwo$ is symmetric if it can be expressed as $g(\|x\|_0)$ for some function $g : [0,n] \to \ftwo$.
We give the following lower bound for symmetric functions:
\begin{theorem}[Lower bound for symmetric functions]
	For any symmetric function $f \colon \ftwo^n \to \ftwo$ that isn't $(1 - \epsilon)$-concentrated on $\{\emptyset , \{1, \dots, n\}\}$:
	$$\distlu{\epsilon/8}(f) \ge \frac{n}{2e}, \quad\quad\quad  \distcu{\Theta(\frac{1 - \epsilon}{n})}(\fplus{}) \ge \frac{n}{2e}. $$
\end{theorem}
\begin{proof}
	First we prove an auxiliary lemma. Let $W_k$ be the set of all vectors in $\ftwo^n$ of Hamming weight $k$.
	\begin{lemma}\label{lem:linear-subspace-hamming-weight-intersection}
		For any $d \in [n/2]$, $k \in [n-1]$ and any $d$-dimensional subspace $\mathcal A_d \le \ftwo^n$:
		\begin{align*}
		\frac{|W_k \cap \mathcal A_d|}{|W_k|} \le \left(\frac{e d}{n}\right)^{min(k, n-k, d)} \le \frac{ed}{n}.
		\end{align*}
	\end{lemma}
	\begin{proof}
		Fix any basis in $\mathcal A_d$ and consider the matrix $M \in \ftwo^{d \times n}$ composed of the basis vectors as rows.
		W.l.o.g we can assume that this matrix is diagonalized and is in the standard form $(I_d, M')$ where $I_d$ is a $d \times d$ identity matrix and $M'$ is a $d \times (n - d)$-matrix.
		Clearly, any linear combination of more than $k$ rows of $M$ has Hamming weight greater than $k$ just from the contribution of the first $d$ coordinates.
		Thus, we have $|W_k \cap \mathcal A_d| \le \sum_{i = 0}^k \binom{d}{i}$.
		
		For any $k \le d$ it is a standard fact about binomials that $\sum_{i = 0}^k \binom{d}{i} \le \left(\frac{e d}{k}\right)^k$.
		On the other hand, we have $|W_k| = \binom{n}{k} \ge (n/k)^k$.
		Thus, we have $\frac{|W_k \cap \mathcal A_d|}{|W_k|} \le \left(\frac{ed}{n}\right)^k$
		and hence for $1 \le k \le d$ the desired inequality holds.
		
		If $d < k$ then consider two cases.
		Since $d \le n/2$ the case $n-d \le k \le n - 1$ is symmetric to $1 \le k \le d$.
		If $d < k < n- d$ then we have $|W_k| > |W_d| \ge (n/d)^d$ and $|W_k \cap \mathcal A_d| \le 2^d$ so that the desired inequality follows.\qedhere
	\end{proof}
	
	Any symmetric function has its spectrum distributed uniformly over Fourier coefficients of any fixed weight.
	Let $w_i = \sum_{S \in W_i} \hat f^2(S)$.
	By the assumption of the theorem we have $\sum_{i = 1}^{n - 1} w_i \ge \epsilon$.
	Thus, by Lemma~\ref{lem:linear-subspace-hamming-weight-intersection} any linear subspace $\mathcal A_d$ of dimension at most $d \le n/2$ satisfies that: 
	\begin{align*}
	\sum_{S \in \mathcal A_d} f^2(S) &\le \hat f^2(\emptyset) + \hat f^2(\{1, \dots, n\}) + \sum_{i = 1}^{n - 1} w_i \frac{|W_i \cap \mathcal A_d|}{|W_i|} \\
	&\le \hat f^2(\emptyset) + \hat f^2(\{1, \dots, n\}) + \sum_{i =1}^{n -1} w_i \frac{ed}{n} \\ 
	&\le (1 - \epsilon) + \epsilon \frac{ed}{n}.
	\end{align*}
	
	Thus, $f$ isn't $1 - \epsilon(1- \frac{ed}{n})$-concentrated on any $d$-dimensional linear subspace, i.e. $\epsu{d}{f} <  1 - \epsilon(1- \frac{ed}{n})$.
	By Part 2 of Theorem \ref{thm:linear-sketch-uniform} this implies that $f$ doesn't have randomized sketches of dimension at most $d$
	which err with probability less than:
	\begin{align*}
	\frac{1}{2} - \frac{\sqrt{1 - \epsilon(1- \frac{ed}{n})}}{2} \ge \frac{\epsilon}{4}\left(1 - \frac{ed}{n}\right) \ge \frac{\epsilon}{8}
	\end{align*}
	where the last inequality follows by the assumption that $d \le \frac{n}{2e}$.
The communication complexity lower bound follows by Corollary~\ref{cor:linear-sketch-uniform-lb} by taking $\theta = \epsilon/8$.

\end{proof}

\section{Turnstile streaming algorithms over $\ftwo$}\label{sec:streaming}
Let $e_i$ be the standard unit vector in $\ftwo^n$.
In the turnstile streaming model the input $x \in \ftwo^n$ is represented as a stream $\sigma = (\sigma_1, \sigma_2, \dots)$ where $\sigma_i \in \{e_1, \dots, e_n\}$.
For a stream $\sigma$ the resulting vector $x$ corresponds to its frequency vector $\fr \sigma \equiv \sum_i \sigma_i$.
Concatenation of two streams $\sigma$ and $\tau$ is denoted as $\sigma \circ \tau$.

\subsection{Random streams}\label{sec:random-streaming}
We consider the following two natural models of random streams over $\ftwo$:

\textbf{Model 1.} In the first model we start with $x \in \ftwo^n$ that is drawn from the uniform  distribution over $\ftwo^n$ and then apply a uniformly random update $y \sim U(\ftwo^n)$ obtaining $x + y$.
In the streaming language this corresponds to a stream $\sigma = \sigma_1 \circ \sigma_2$ where $\fr \sigma_1  \sim U(\ftwo^n)$ and $ \fr \sigma_2 \sim U(\ftwo^n)$.
A specific example of such stream would be one where for both $\sigma_1$ and $\sigma_2$ we flip an unbiased coin to decide whether or not to include a vector $e_i$ in the stream for each value of $i$.
The expected length of the stream in this case is $n$.

\textbf{Model 2.} In the second model we consider a stream $\sigma$ which consists of uniformly random updates.
Let $\sigma_i = e_{r(i)}$ where $r(i) \sim U([n])$.
This corresponds to each update being a flip in a coordinate of $x$ chosen uniformly at random.
This model is equivalent to the previous model but requires longer streams to mix.
Using coupon collector's argument such streams of length $\Theta(n \log n)$ can be divided into two substreams $\sigma_1$ and $\sigma_2$ such that with high probability both $\fr \sigma_1$ and $\fr \sigma_2$ are uniformly distributed over $\ftwo^n$ and $\sigma = \sigma_1 \circ \sigma_2$.

\begin{theorem}\label{thm:random-streaming}
Let $f \colon \ftwo^n \to \ftwo$ be an arbitrary function. 
In the two random streaming models for generating $\sigma$ described above any algorithm that computes $f(\fr \sigma)$ with probability at least $1 - \Theta(1/n)$ in the end of the stream has to use space that is at least $\distlu{1/3}(f)$.
\end{theorem}
\begin{proof}
The proof follows directly from Theorem~\ref{thm:linear-sketch-uniform-main} as in both models we can partition the stream into $\sigma_1$ and $\sigma_2$ such that $\fr \sigma_1$ and $\fr \sigma_2$ are both distributed uniformly over $\ftwo^n$.
We treat these two frequency vectors as inputs of Alice and Bob in the communication game.
Since communication $\distcu{\Theta(1/n)}(\fplus{}) \ge \distlu{1/3}(f)$ is required no streaming algorithm with less space exists as otherwise Alice would transfer its state to Bob with less communication.
\end{proof}

\subsection{Adversarial streams}\label{sec:adversarial-streaming}
We  now show that any randomized turnstile streaming algorithm for computing $f : \ftwo^n \to \ftwo$ with error probability $\delta$ has to use space that is at least $\rl{6\delta}(f) - O(\log n + \log(1/\delta))$ under adversarial sequences of updates.
The proof is based on the recent line of work that shows that this relationship holds for real-valued sketches~\cite{G08,LNW14,AHLW16}.
The proof framework developed by~\cite{G08,LNW14,AHLW16} for real-valued sketches consists of two steps. First, a turnstile streaming algorithm is converted into a path-independent stream automaton (Definition~\ref{def:path-ind-automaton}). Second, using the theory of modules and their representations it is shown that such automata can always be represented as linear sketches.
We observe that the first step of this framework can be left unchanged under $\ftwo$.
However, as we show the second step can be significantly simplified as path-independent automata over $\ftwo$ can be directly seen as linear sketches without using module theory.
Furthermore, since we are working over $\ftwo$ we also avoid the $O(\log m)$ factor loss in the reduction between path independent automata and linear sketches that is present in~\cite{G08}.

We use the following abstraction of a \textit{stream automaton} from~\cite{G08,LNW14,AHLW16} adapted to our context to represent general turnstile streaming algorithms over $\ftwo$.
\begin{definition}[Deterministic Stream Automaton]\label{def:stream-automaton}
A \emph{deterministic stream automaton} $\cA$ is a Turing machine that uses two tapes, an undirectional read-only input tape and a bidirectional work tape.
The input tape contains the input stream $\sigma$.
After processing the input, the automaton writes an output, denoted as $\phi_\cA(\sigma)$, on the work tape.
A configuration (or state) of $\cA$ is determined by the state of its finite control, head position, and contents of the work tape.
The computation of $\cA$ can be described by a transition function $\oplus_\cA : C \times \ftwo \to C$, where $C$ is the set of all possible configurations.
For a configuration $c \in C$ and a stream $\sigma$, we denote by $c \oplus_\cA \sigma$ the configuration of $\cA$ after processing $\sigma$ starting from the initial configuration $c$.
The set of all configurations of $\cA$ that are reachable via processing some input stream $\sigma$ is denoted as $C(\cA)$.
The space of $\cA$ is defined as $\cS(\cA) = \log |C(\cA)|$.
\end{definition}

We say that a deterministic stream automaton computes a function $f:\ftwo^n \to \ftwo$ over a distribution $\Pi$ if $\Pr_{\sigma \sim \Pi}[\phi_\cA(\sigma) = f(\fr \sigma)] \ge 1- \delta$.
\begin{definition}[Path-independent automaton]\label{def:path-ind-automaton}
An automaton $\cA$ is said to be \emph{path-independent} if for any configuration $c$ and any input stream $\sigma$, $c \oplus_\cA \sigma$ depends only on $\fr \sigma$ and $c$.
\end{definition}

\begin{definition}[Randomized Stream Automaton]
A \emph{randomized stream automaton}  $\cA$ is a deterministic automaton with an additional tape for the random bits.
This random tape is initialized with a random bit string $R$ before the automaton is executed.
During the execution of the automaton this bit string is used in a bidirectional read-only manner while the rest of the execution is the same as in the deterministic case.
A randomized automaton $\cA$ is said to be path-independent if for each possible fixing of its randomness $R$ the deterministic automaton $\cA_R$ is path-independent.
The space complexity of $\cA$ is defined as $\cS(\cA) = \max_R (|R| + \cS(\cA_R))$.
\end{definition}

Theorems 5 and 9 of~\cite{LNW14} combined with the observation in Appendix A of~\cite{AHLW16} that guarantees path independence yields the following:
\begin{theorem}[Theorems 5 and 9 in~\cite{LNW14} +~\cite{AHLW16}]
Suppose that a randomized stream automaton $\cA$ computes $f$ on any stream with probability at least $1 - \delta$. For an arbitrary distribution $\Pi$ over streams there exists a deterministic\footnote{We note that~\cite{LNW14} construct $\cB$ as a randomized automaton in their Theorem 9 but it can always be made deterministic by fixing the randomness that achieves the smallest error.}  path independent stream automaton $\cB$ that computes $f$ with probability $1 - 6 \delta$ over $\Pi$ such that $\cS(\cB) \le \cS(\cA) + O(\log n +\log (1/\delta))$.
\end{theorem}

The rest of the argument below is based on the work of Ganguly~\cite{G08} adopted for our needs.
Since we are working over a finite field we also avoid the $O(\log m)$ factor loss in the reduction between path independent automata and linear sketches that is present in Ganguly's work.

Let $A_n$ be a path-independent stream automaton over $\ftwo$ and let $\oplus$ abbreviate $\oplus_{A_n}$.
%Let $C(A_n)$ denote the space of all its configurations.
Define the function $\ast: \ftwo^n\times C(A_n) \rightarrow C(A_n)$ as: $x \ast a = a \oplus \sigma\text{, where } freq(\sigma)=x.$
Let $o$ be the initial configuration of $A_n$.
The \textit{kernel}  $M_{A_n}$ of $A_n$ is defined as $M_{A_n} = \{x \in \ftwo^n : x \ast o = 0^n \ast o\}$.
\begin{proposition}
The kernel $M_{A_n}$ of a path-independent automaton $A_n$ is a linear subspace of $\ftwo^n$.
\end{proposition}
\begin{proof}
For $x,y \in M_{A_n}$ by path independence $(x + y) \ast o = x \ast (y \ast o) = 0^n \ast o$ so $x + y \in M_{A_n}$.
\end{proof}
Since $M_{A_n} \le \ftwo^n$ the kernel partitions $\ftwo^n$ into cosets of the form $x + M_{A_n}$.
Next we show that there is a one to one mapping between these cosets and the states of $A_n$.

\begin{proposition}
For $x, y \in \ftwo^n$ and a path independent automaton $A_n$ with a kernel $M_{A_n}$ it holds that $x \ast o = y \ast o$ if and only if $x$ and $y$ lie in the same coset of $M_{A_n}$.
\end{proposition}
\begin{proof}
By path independence $x \ast o = y \ast o$ iff $x \ast (x \ast o) = x \ast (y \ast o)$ or equivalently $0^n \ast o = (x + y) \ast o$. The latter condition holds iff $x + y \in M_{A_n}$ which is equivalent to $x$ and $y$ lying in the same cost of $M_{A_n}$.
\end{proof}
The same argument implies that the the transition function of a path-independent automaton has to be linear since $(x + y) \ast o = x \ast (y \ast o)$. Combining these facts together we conclude that a path-independent automaton has at least as many states as the best  deterministic $\ftwo$-sketch for $f$ that succeeds with probability at least $1 - 6\delta$ over $\Pi$ (and hence the best randomized sketch as well). Putting things together we get:

\begin{theorem}\label{thm:adversarial-streaming}
	Any randomized streaming algorithm that computes $f:\ftwo^n \to \ftwo$ under arbitrary updates over $\ftwo$ with error probability at least $1 - \delta$ has space complexity at least $\rl{6\delta}(f) - O(\log n + \log(1/\delta))$.
\end{theorem}

\section{Linear threshold functions}\label{sec:ltf}
In this section it will be convenient to represent the domain as $\{0,1\}^n$ rather than $\ftwo^n$.
We define the sign function $sign(x)$ to be $1$ if $x \ge 0$ and $0$ otherwise.
\begin{definition}\label{def:ltf}
	A monotone linear threshold function (LTF) $f \colon \{0,1\} \rightarrow \oo$ is defined by a collection of weights $w_1 \ge w_2 \dots \ge w_n \ge 0$ as follows:
	\begin{align*}
	f(x_1, \dots, x_n) = sign\left(\sum_{i = 1}^n w_i x_i - \theta\right),
	\end{align*}
	where $\theta$ is called the \em{threshold} of the LTF.
	The \em{margin} of the LTF is defined as: 
	\begin{align*}
	m = \min_{x \in \{0,1\}^n} \left|\sum_{i = 1}^n w_i x_i - \theta\right|.
	\end{align*}
\end{definition}

W.l.o.g we can assume that LTFs normalized so that $\sum_{i = 1}^n w_i = 1$. The monotonicity in the above definition is also without loss of generality as for negative weights we can achieve monotonicity by complementing individual bits.

\begin{theorem}~\cite{MO09}\label{thm:mo10}
	There is a randomized linear sketch for LTFs of size $O(\left(\frac{\theta}{m}\right)^2)$.
\end{theorem}

Below we prove the following conjecture.

\begin{conj}~\cite{MO09}\label{conj:mo10}
	There is a randomized linear sketch for LTFs of size $O\left(\frac{\theta}{m} \log \left(\frac{\theta}{m}\right)\right)$.
\end{conj}

In fact, all weights which are below the margin can be completely ignored when evaluating the LTF.
\begin{lemma}\label{lem:ltf-margin}
	Let $f$ be a monotone LTF with weights $w_1 \ge w_2 \ge \dots \ge w_n$, threshold $\theta$ and margin $m$.
	Let $f^{\ge 2m}$ be an LTF with the same threshold and margin but only restricted to weights $w_1 \ge w_2 \ge \dots \ge w_t$, where $t$ is the largest integer such that $w_t \ge 2m$.
	Then $f = f^{\ge m}$.
\end{lemma}
\begin{proof}
	For the sake of contradiction assume there exists an input $(x_1, \dots, x_n)$ such that $f(x_1, \dots, x_n) = 1$ while  $f^{\ge 2m}(x_1, \dots, x_t) = 0$.
	Fix the largest $t^* \ge t$ such that $sign\left(\sum_{i = 1}^{t^*} w_i x_i- \theta\right) = 0$ while $sign\left(\sum_{i = 1}^{t^* + 1} w_i x_i - \theta\right) = 1$. Clearly $w_{t^*+1} \ge 2m$, a contradiction.
\end{proof}

The above lemma implies that after dropping the weights which are below $2m$ together with the corresponding variables and reducing the value of $n$ accordingly we can also make the margin equal to $w_n/2$.
This observation also gives the following straightforward corollary that proves Conjecture~\ref{conj:mo10} about LTFs (up to a logarithmic factor in $n$).

\begin{corollary}\label{cor:naive-sketching}
	There is a randomized linear sketch for LTFs of size $O\left(\frac{\theta}{m} \log n\right)$.
\end{corollary}
\begin{proof}
	We will give a bound on $|\{x \colon f(x) = 0\}|$.
	If $f(x) = 0$ then $\sum_{i = 1}^n w_i x_i < \theta$.
	Since all weights are at least $w_n$ the total number of such inputs is at most $\binom{n}{\theta/w_n} = \binom{n}{\theta/ 2m} \le (n + 1)^{\theta/2m}$.
	Thus applying the random $\ftwo$-sketching bound (Fact~\ref{prop:l0-bound}) we get a sketch of size $O\left(\frac{\theta }{ m} \log n\right)$ as desired.
\end{proof}

Combined with Theorem~\ref{thm:mo10}
the above corollary proves Conjecture~\ref{conj:mo10} except in the case when $\beta \log \left(\theta/ m\right) < \theta/m < n^{\alpha}$ for all $\alpha > 0$ and $\beta < \infty$.
This matches the result of~\cite{LZ13}.

A full proof of Conjecture~\ref{conj:mo10} can be obtained by using hashing to reduce the size of the domain from $n$ down to $poly(\theta /m)$.

\begin{theorem}\label{thm:naive-sketching+hashing}
		There is a randomized linear sketch for LTFs of size $O\left(\frac{\theta}{m} \log \left(\frac{\theta}{m}\right)\right)$ that succeeds with any constant probability.
\end{theorem}
\begin{proof}
	It suffices to only consider the case when $\theta/m >100$ since otherwise the bound follows trivially from Theorem~\ref{thm:mo10}. Consider computing a single linear sketch $\sum_{i \in S} x_i$ where $S$ is a random vector in $\ftwo^n$ with each coordinate set to $1$ independently with probability $10 m^2 / \theta^2$.
	This sketch lets us distinguish the two cases $\|x\|_0 > \theta^2 / m^2$ vs. $\|x\|_0 \le \theta/m$ with constant probability. Indeed:
	
	\textbf{Case 1.} $\|x\|_0 > \theta^2 / m^2$. The probability that a set $S$ contains a non-zero coordinate of $x$ in this case is at least:
	$$1 - \left(1 - \frac{10 m^2}{\theta^2}\right)^{\frac{\theta^2}{m^2}} \ge 1 - (1/e)^{10} > 0.9$$
	Conditioned on this event the parity evaluate to $1$ with probability at least $1/2$.
	Hence, overall in this case the parity evaluates to $1$ with probability at least $0.4$.
	
	\textbf{Case 2.} $\|x\|_0 \le \theta/m$. In this case this probability that $S$ contains a non-zero coordinate and hence the parity can evaluate to $1$ is at most: 
	$$1 - \left(1 - \frac{10 m^2}{\theta^2}\right)^{\theta/m} < 1 - \left(1/2e\right)^{1/10} < 0.2$$
	
	Thus, a constant number of such sketches allows to distinguish the two cases above with constant probability. 
	If the test above declares that $\|x\|_0 > \theta^2/m^2$ then we output $1$ and terminate.
	Note that conditioned on the test above being correct it never declares that $\|x\|_0 > \theta^2/m^2$ while $\|x\|_0 \le \theta/m$.
	Indeed in all such cases, i.e. when $\|x\|_0 > \theta/m$ we can output $1$ since if $\|x\|_0 > \theta/m$ then $\sum_{i = 1}^n w_i x_i \ge \|x\|_0 w_n \ge \frac{\theta w_n}{m} = 2 \theta$, where we used the fact that by Lemma~\ref{lem:ltf-margin} we can set $m = w_n/2$.

	For the rest of the proof we thus condition on the event that $\|x\|_0 \le \theta^2 /m^2$.
	By hashing the domain $[n]$ randomly into $O\left(\theta^4/m^4\right)$ buckets we can ensure that no non-zero entries of $x$ collide with any constant probability that is arbitrarily close to $1$. This reduces the input length from $n$ down to $O\left(\theta^4/m^4\right)$ and we can apply Corollary~\ref{cor:naive-sketching} to complete the proof.
	\footnote{We note that random hashing doesn't interfere with the linearity of the sketch as it corresponds to treating collections of variables that have the same hash as a single variable representing their sum over $\ftwo$. Assuming no collisions this sum evaluates to $1$ if and only if a variable of interest is present in the collection.}
\end{proof}

This result is also tight as follows from the result of Dasgupta, Kumar and Sivakumar~\cite{DKS12} discussed in the introduction.
Consider the Hamming weight function $Ham_{\ge d}(x) \equiv \|x\|_0 \ge d$. This function satisfies $\theta = d/n$, $m = 1/2n$. A straightforward reduction from small set disjointness shows that the one-way communication complexity of the XOR-function $Ham_{\ge d}(x \oplus y)$ is $\Omega(d \log d)$. This shows that the bound in Theorem~\ref{thm:naive-sketching+hashing} can't be improved without any further assumptions about the LTF.

\section{Towards the proof of Conjecture~\ref{conj:main}}\label{sec:one-bit-lb}\todo{add some intro remarks}

We call a function $f:\ftwo^n \to \oo$ \textit{non-linear} if for all $S \in \ftwo^n$ there exists $x \in \ftwo^n$ such that $f(x) \neq \chi_S(x)$.
Furthermore, we say that $f$ is $\eps$-far from being linear if: 
$$\max_{S \in \ftwo^n}\left[\Pr_{x \sim U(\ftwo^n)}[\chi_S(x) = f(x)]\right] = 1 - \eps.$$

The following theorem is our first step towards resolving Conjecture~\ref{conj:main}. Since non-linear functions don't admit $1$-bit linear sketches we show that the same is also true for the corresponding communication complexity problem, namely no $1$-bit communication protocol for such functions can succeed with a small constant error probability.
\begin{theorem}
	For any non-linear function $f$ that is at most $1/10$-far from linear $\distcm{1/200}{(\fplus{})} > 1$.
\end{theorem}
\begin{proof}
	Let $S = \argmax_T \left[\Pr_{x \in \ftwo^n}[\chi_T(x) = f(x)\right]$.
	Pick $z \in \ftwo^n$ such that $f(z) \neq \chi_S(z)$.
	Let the distribution over the inputs $(x,y)$ be as follows:
	 $y \sim U(\ftwo^n)$ and $x \sim \mathcal D_y$ where $D_y$ is defined as:
	\begin{align*}
	D_y = \begin{cases}
	&y + z \text { with probability } 1/2, \\
	&U(\ftwo^n) \text{ with probability } 1/2.
	\end{cases}
	\end{align*}
	Fix any deterministic Boolean function $M(x)$ that is used by Alice to send a one-bit message based on her input.
	For a fixed Bob's input $y$ he outputs $g_y(M(x))$ for some function $g_y$ that can depend on $y$.
	Thus, the error that Bob makes at predicting $f$ for fixed $y$ is at least:
	\begin{align*}
	\frac{1 - \left|\mathbb E_{x \sim D_y} \left[g_y(M(x)) f(x + y)\right]\right|}{2}.
	\end{align*}	\todo{Do we actually need an absolute value here?}
	The key observation is that since Bob only receives a single bit message there are only four possible functions $g_y$ to consider for each $y$: constants $-1/1$ and $\pm M(x)$.
	
	\paragraph{Bounding error for constant estimators.} For both constant functions we introduce notation $B^{c}_y = \left|\mathbb E_{x \sim D_y} \left[g_y(M(x)) f(x + y)\right]\right|$ and have: 
	\begin{align*}
	B^c_y &= \left|\mathbb E_{x \sim D_y} \left[g_y(M(x)) f(x + y)\right]\right| 
	= |\mathbb E_{x \sim D_y}[f(x + y)]| = \left|\frac12 f(z) + \frac12 \mathbb E_{w \sim U(\ftwo^n)} [f(w)]\right|
	\end{align*}
	
	If $\chi_S$ is not constant then $\left|\mathbb E_{w \sim U(\ftwo^n)} [f(w)]\right| \le 2\epsilon$ we have:
	\begin{align*}
	\left|\frac12 f(z) + \frac12 \mathbb E_{w \sim U(\ftwo^n)} [f(w)]\right| \le \frac12 \left(|f(z)| + \left|\mathbb E_{w \sim U(\ftwo^n)} [f(w)]\right|\right)\le 1/2 + \epsilon.
	\end{align*}
	
	If $\chi_S$ is a constant then w.l.o.g $\chi_S$ = 1 and $f(z) = -1$. Also $\mathbb E_{w \sim U(\ftwo^n)}[f(w)] \ge 1-2\epsilon$. Hence we have:
	
	\begin{align*}
	\left|\frac12 f(z) + \frac12 \mathbb E_{w \sim U(\ftwo^n)} [f(w)]\right| = \frac12 \left|-1 + \mathbb E_{w \sim U(\ftwo^n)} [f(w)]\right| \le \epsilon.
	\end{align*}
	Since $\epsilon \le 1/10$ in both cases $B^c_y \le \frac{1}{2} + \epsilon$ which is the bound we will use below.
	
	\paragraph{Bounding error for message-based estimators.}
	For functions $\pm M(x)$ we need to bound $\left|\mathbb E_{x \sim D_y} \left[M(x) f(x + y)\right]\right|$.
	We denote this expression as $B^M_y$.
	Proposition~\ref{prop:message-bound} shows that $\mathbb E_y[B^M_y] \le \frac{\sqrt{2}}{2}\left(1 + \epsilon\right)$.

\begin{proposition}\label{prop:message-bound}
	$\mathbb E_{y \sim U(\ftwo^n)}\left[\left|\mathbb E_{x \sim D_y} \left[M(x) f(x + y)\right]\right|\right] \le \frac{\sqrt{2}}{2}\left(1 + \epsilon\right)$.
\end{proposition}
We have:
\begin{align*}
&\mathbb E_{y}\left[\left|\mathbb E_{x \sim D_y} \left[M(x) f(x + y)\right]\right|\right] \\
& = \mathbb E_{y}\left[\left|\frac{1}{2} \left( M(y + z) f(z) + \mathbb E_{x \sim D_y}[M(x) f(x + y)]\right)\right|\right] \\
& = \frac{1}{2} \mathbb E_{y}\left[\left|\left( M(y + z) f(z) + (M*f)(y)\right)\right|\right] \\
& \le\frac{1}{2}  \left(\mathbb E_{y}\left[\left( \left( M(y + z) f(z) + (M*f)(y)\right)\right)^2\right]\right)^{1/2} \\
& =\frac{1}{2}  \left(\mathbb E_{y}\left[ \left( (M(y + z) f(z))^2 + ((M*f)(y))^2 + 2 M(y + z) f(z) (M * f)(y))\right)\right]\right)^{1/2} \\
& =\frac{1}{2}  \left(\mathbb E_{y}\left[ \left( (M(y + z) f(z))^2\right] + \mathbb E_{y}\left[((M*f)(y))^2\right] + 2 \mathbb E_{y}\left[M(y + z) f(z) (M * f)(y))\right)\right]\right)^{1/2} 
\end{align*}

We have $(M(y+ z) f(z))^2 = 1$ and also by Parseval, expression for the Fourier spectrum of convolution and Cauchy-Schwarz:
\begin{align*}
\mathbb E_y[((M * f)(y))^2]  
= \sum_{S \in \ftwo^n} \widehat {M * f}(S)^2 
= \sum_{S \in \ftwo^n} \widehat {M}(S)^2 \hat{f}(S)^2 
\le ||M||_2 ||f||_2 = 1
\end{align*}

Thus, it suffices to give a bound on $\mathbb E[M(y + z) f(z) (M * f)(y))]$.
First we give a bound on $(M* f)(y)$:
\begin{align*}
(M* f)(y) = \mathbb E_x[M(x) f(x + y)]  \le \mathbb E_x[M(x)\chi_S(x + y)]  + 2\epsilon  \\
\end{align*}
Plugging this in we have:
\begin{align*}
&\mathbb E_y[M(y + z) f(z) (M * f)(y))] \\
&= - \chi_S(z)\mathbb E_y[M(y + z)(M * f)(y))] \\
& \le - \chi_S(z)\mathbb E_y\left[M(y + z) (M * \chi_S)(y)\right] + 2 \epsilon  \\
& =  - \chi_S(z) (M * (M * \chi_S))(z) + 2 \epsilon \\
& = - \chi_S(z)^2 \hat{M}(S)^2 + 2 \epsilon \\
& \le 2 \epsilon.
\end{align*}
where we used the fact that the Fourier spectrum of $(M * (M * \chi_S))$ is supported on $S$ only and $\widehat{M * (M * \chi_S)}(S) = \hat{M}^2(S)$ and thus $ (M * (M * \chi_S))(z) = \hat{M}^2(S)\chi_S(z)$.

Thus, overall, we have: 
\begin{align*}
\mathbb E_{y}\left[\left|\mathbb E_{x \sim D_y} \left[M(x) f(x + y)\right]\right|\right] \le 
\frac{1}{2} \sqrt{2 + 4\epsilon} \le \frac{\sqrt{2}}{2}(1 + \epsilon).\qedhere
\end{align*}
\end{proof}

	\paragraph{Putting things together.}
	We have that the error that Bob makes is at least:
	\begin{align*}
	&\mathbb E_y\left[\frac{1 - max(B^c_y,B^M_y)}{2}\right]  = \frac{1 - \mathbb E_y[max(B^c_y,B^M_y)]}{2} 
	\end{align*}
	Below we now bound $\mathbb E_y[max(B^c_y,B^M_y)]$ from above by $99/100$ which shows that the error is at least $1/200$.
	\begin{align*}
	&\mathbb E_y[max(B^c_y,B^M_y)] \\
	& = \Pr[B^M_y \ge 1/2 + \epsilon] \mathbb E[B^M_y | B^M_y \ge 1/2 + \epsilon] + Pr[B^M_y < 1/2 + \epsilon] \left(\frac{1}{2} + \epsilon\right) \\
	& = \mathbb E_y[B^M_y] + Pr[B^M_y < 1/2 + \epsilon] \left(\frac{1}{2} + \epsilon -\mathbb E[B^M_y | B^M_y < 1/2 + \epsilon] \right)
	\end{align*}
	Let $\delta = Pr[B^M_y < 1/2 + \epsilon]$.
	Then the first of the expressions above gives the following bound:
	\begin{align*}
	\mathbb E_y[max(B^c_y,B^M_y)] \le (1 - \delta) + \delta \left(\frac{1}{2} + \epsilon\right) = 1 - \frac{\delta}{2} + \epsilon \delta  \le 1 - \frac{\delta}{2} + \epsilon
	\end{align*}
	The second expression gives the following bound:
	\begin{align*}
	\mathbb E_y[max(B^c_y,B^M_y)] \le \frac{\sqrt{2}}{2}\left(1 + \epsilon\right) + \delta\left(\frac{1}{2} + \epsilon\right) \le \frac{\sqrt{2}}{2} + \frac{\delta}{2} + \frac{\sqrt{2}}{2} \epsilon + \epsilon.
	\end{align*}
	These two bounds are equal for $\delta = 1 - \frac{\sqrt{2}}{2} \left(1 + \epsilon\right)$ and hence the best of the two bounds is always at most $(\frac{\sqrt{2}}{4} + \frac{1}{2}) + \epsilon \left(\frac{\sqrt{2}}{4} + 1\right) \le \frac{99}{100}$ where the last inequality uses the fact that $\epsilon \le \frac{1}{10}$.

\bibliographystyle{alpha}
\bibliography{../linsketch}

\appendix
\section*{Appendix}
\addcontentsline{toc}{section}{Appendix}
\renewcommand{\thesection}{\Alph{section}}
\setcounter{theorem}{0}

\section{Deterministic $\ftwo$-sketching}\label{app:deterministic}

%Recall that we denote the deterministic sketching complexity of $f$ as $\ds(f)$.
In the deterministic case it will be convenient to represent $\ftwo$-sketch of a function $f \colon \ftwo^n \to \ftwo$ as a $d \times n$ matrix $M_f \in \ftwo^{d \times n}$ that we call the \textit{sketch matrix}. The $d$ rows of $M_f$ correspond to vectors $\alpha_1, \dots, \alpha_d$ used in the deterministic sketch so that the sketch can be computed as $M_f x$. W.l.o.g below we will assume that the sketch matrix $M_f$ has linearly independent rows and that the number of rows in it is the smallest possible among all sketch matrices (ties in the choice of the sketch matrix are broken arbitrarily).

The following fact is standard (see e.g.~\cite{MO09,GOSSW11}):
\begin{fact}\label{fact:fourier-dimension}
	For any function $f \colon \ftwo^n \to \ftwo$ it holds that $\ds(f) = dim(f) = rank(M_f)$.
Moreover, set of rows of $M_f$ forms a basis for a subspace $A\le \ftwo^n$ containing all non-zero coefficients of $f$.
\end{fact}

\subsection{Disperser argument}\label{sec:disperser}

We show that the following basic relationship holds between deterministic linear sketching complexity and the property of being an affine disperser.
For randomized $\ftwo$-sketching an analogous statement holds for affine extractors as shown in Lemma~\ref{lem:rand-sketch-extractors}.

\begin{definition}[Affine disperser]
	A function $f$ is an affine disperser of dimension at least $d$ if for any affine subspace of $\ftwo^n$ of dimension at least $d$ the restriction of $f$ on it is a non-constant function.
\end{definition}

\begin{lemma}
	Any  function $f \colon \ftwo^n \rightarrow \ftwo$ which is an  affine disperser of dimension at least $d$ has deterministic linear sketching complexity at least $n-d + 1$.
\end{lemma}
\begin{proof}
	Assume for the sake of contradiction that there exists a linear sketch matrix $M_f$ with $k \le n - d$ rows and a deterministic function $g$ such that $g(M_f x) = f(x)$ for every $x \in \ftwo^n$.
	For any vector $b \in \ftwo^k$, which is in the span of the columns of $M_f$, the set of vectors $x$ which satisfy $M_f x = b$ forms an affine subspace of dimension at least $n - k \ge d$.
	Since $f$ is an affine disperser for dimension at least $d$ the restriction of $f$ on this subspace is non-constant. 
	However, the function $g(M_f x) = g(b)$ is constant on this subspace and thus there exists $x$ such that $g(M_f x) \neq f(x)$, a contradiction.
\end{proof}

\subsection{Composition and convolution}\label{sec:det-composition-convolution}
In order to prove a composition theorem for $\dl$ we introduce the following operation on matrices which for a lack of a better term we call matrix super-slam\footnote{This name was suggested by Chris Ramsey.}.

\begin{definition}[Matrix super-slam]
	For two matrices $A \in \mathbb F_2^{a \times n}$ and $B \in \mathbb F_2^{b \times m}$ their \textit{super-slam} $A \dagger B \in \mathbb F_2^{a b^n \times nm}$ is a block matrix consisting of $a$ blocks $(A \dagger B)_i$. The $i$-th block $(A \dagger B)_i \in \mathbb F_2^{b^n \times nm}$ is constructed as follows: for every vector $j \in \{1, \dots, b\}^n$ the corresponding row of $(A \dagger B)_i$ is defined as $(A_{i,1}B_{j_1}, A_{i,2} B_{j_2}, \dots, A_{i,n}B_{j_n})$, where $B_k$ denotes the $k^{th}$ row of $B$.
\end{definition}
\begin{proposition}\label{prop:super-slam-rank}
	$rank(A \dagger B) \ge rank(A) rank(B)$.
\end{proposition}
\begin{proof}
	Consider the matrix $C$ which is a subset of rows of $A \dagger B$ where from each block $(A \dagger B)_i$ we select only $b$ rows corresponding to the vectors $j$ of the form $\alpha^n$ for all $\alpha \in \{1, \dots, b\}$.
	Note that $C \in \mathbb F_2^{ab \times mn}$ and $C_{(i,k),(j,l)} = A_{i,j} B_{k,l}$.
	Hence, $C$ is a Kronecker product of $A$ and $B$ and we have:
	\begin{align*}
	rank(A \dagger B) \ge rank(C) = rank(A)rank(B).\qedhere
	\end{align*} 
\end{proof}

The following composition theorem for $\dl$ holds as long as the inner function is balanced:

\begin{lemma}\label{lem:det-sketch-composition}
	For $f \colon \ftwo^n \rightarrow \ftwo$ and $g \colon \ftwo^m \rightarrow \ftwo$ if $g$ is a balanced function then:
	\begin{align*}
	\ds (f \circ g) \ge \ds (f) \ds (g)
	\end{align*}
\end{lemma}

\begin{proof}
	The multilinear expansions of $f$ and $g$ are given as $f(y)= \sum_{S \in \ftwo^n} \hat f(S) \chi_S(y)$ and $g(y)= \sum_{S \in \ftwo^m} \hat g(S) \chi_S(y)$.
	The multilinear expansion of $f \circ g$ can be obtained as follows.
	For each monomial $\hat f(S) \chi_S(y)$ in the multilinear expansion of $f$ and each variable $y_i$ substitute $y_i$ by the multilinear expansion of $g$ on a set of variables $x_{ m (i - 1) + 1, \dots, m i}$.
	Multiplying all these multilinear expansions corresponding to the term $\hat f(S)\chi_S$ gives a polynomial which is a sum of at most $b^n$ monomials where $b$ is the number of non-zero Fourier coefficients of $g$. Each such monomial is obtained by picking one monomial from the multilinear expansions corresponding to different variables in $\chi_S$ and multiplying them.
	Note that there are no cancellations between the monomials corresponding to a fixed $\chi_S$.
	Moreover, since $g$ is balanced and thus $\hat g(\emptyset) = 0$ all monomials corresponding to different characters $\chi_S$ and $\chi_{S'}$ are unique since $S$ and $S'$ differ on some variable and substitution of $g$ into that variable doesn't have a constant term but introduces new variables. Thus, the characteristic vectors of non-zero Fourier coefficients of $f \circ g$ are the same as the set of rows of the super-slam of the sketch matrices $M_f$ and $M_g$ (note, that in the super-slam some rows can be repeated multiple times but after removing duplicates the set of rows of the super-slam and the set of characteristic vectors of non-zero Fourier coefficients of $f \circ g$ are exactly the same).
	Using Proposition~\ref{prop:super-slam-rank} and Fact~\ref{fact:fourier-dimension} we have: 
	\begin{align*}
	\ds (f \circ g) = rank(M_{f \circ g}) = rank(M_f \dagger M_g) \ge rank(M_f) rank(M_g) = \ds (f) \ds (g). \qedhere
	\end{align*}
\end{proof}

Deterministic $\ftwo$-sketch complexity of convolution satisfies the following property:
\begin{proposition}\label{prop:det-sketch-convolution}
	$\ds (f * g) \le \min(\ds (f), \ds (g)).$
\end{proposition}
\begin{proof}
	The Fourier spectrum of convolution is given as $\widehat {f * g}(S) = \hat f(S) \hat g(S)$.
	Hence, the set of non-zero Fourier coefficients of $f * g$ is the intersection of the sets of non-zero coefficients of $f$ and $g$. Thus by Fact~\ref{fact:fourier-dimension} 
	we have $\dl(f * g) \le \min(rank(M_f, M_g)) = \min(\dl(f), \dl(g))$.
\end{proof}

\section{Randomized $\ftwo$-sketching}\label{app:randomized}
We represent randomized $\ftwo$-sketches as distributions over $d \times n$ matrices over $\ftwo$. For a fixed such distribution $\mathcal M_f$ the randomized sketch is computed as $\mathcal M_f x$. If the set of rows of $\mathcal M_f$ satisfies Definition~\ref{def:rand-f2-sketch} for some reconstruction function $g$ then we call it a \textit{randomized sketch matrix} for $f$.

\subsection{Extractor argument}\label{sec:extractor}

We now establish a connection between randomized $\ftwo$-sketching and affine extractors which will be used to show that the converse of Part 1 of Theorem~\ref{thm:linear-sketch-uniform} doesn't hold for arbitrary distributions.
\begin{definition}[Affine extractor]
	A function $f:\ftwo^n \to \ftwo$ is an affine $\delta$-extractor if for any affine subspace $A$ of $\ftwo^n$ of dimension at least $d$ it satisfies:
	$$\min_{z \in \{0,1\}} \Pr_{x \sim U(A)}[ f(x) = z] > \delta.$$
\end{definition}

\begin{lemma}\label{lem:rand-sketch-extractors}
	For any $f \colon \ftwo^n \rightarrow \ftwo$ which is an affine $\delta$-extractor of dimension at least $d$ it holds that:
	$$\rl{\delta}(f) \ge n - d + 1.$$
\end{lemma}
\begin{proof}
	For the sake of contradiction assume that there exists a randomized linear sketch with a reconstruction function $g : \ftwo^k \to \ftwo$ and a randomized sketch matrix $\mathcal M_f$ which is a distribution over matrices with $k \le n-d$ rows.
	First, we show that:
		\begin{align*}
		\Pr_{x \sim U(\ftwo^n) M \sim \mathcal M_f}\left[g(Mx) \neq f(x)\right] > \delta.
		\end{align*}
		Indeed, fix any matrix $M \in  supp (\mathcal M_f)$.
		For any affine subspace $\mathcal S$ of the form $\mathcal S = \{x \in \ftwo^n | Mx = b\}$ of dimension at least $n - k \ge d$ we have that $\min_{z \in \{0,1\}} \Pr_{x \sim U(\mathcal S)}[ f(x) = z] > \delta $.
		This implies that $\Pr_{x \sim U(\mathcal S)}[ f(x) \neq g(Mx)] > \delta $.
		Summing over all subspaces corresponding to the  fixed $M$ and all possible choices of $b$ we have that 
		$\Pr_{x \sim U(\ftwo^n)}[ f(x) \neq g(Mx)] > \delta$.
		Since this holds for any fixed $M$ the bound follows.
	
	Using the above observation it follows by averaging over $x \in \{0,1\}^n$
	that there exists $x^* \in \{0,1\}^n$ such that: 
	\begin{align*}
	\Pr_{M \sim \mathcal M_f}\left[g(Mx^*) \neq f(x^*)\right] > \delta.
	\end{align*}
	This contradicts the assumption that $\mathcal M_f$ and $g$ form a randomized linear sketch of dimension $k \le n - d$.\qedhere
\end{proof}

\begin{fact}
	The inner product function $IP(x_1, \dots x_n)$ = $\sum_{i = 1} ^{n/2} x_{2 i - 1} \wedge x_{2i}$ is an $(1/2 - \epsilon)$-extractor for affine subspaces of dimension $\ge (1/2 + \alpha) n$ where $\epsilon = \exp(-\alpha n)$. 
	
\end{fact}

\begin{corollary}\label{cor:ip-lower-bound}
	Randomized linear sketching complexity of the inner product function is at least $n/2 - O(1)$.
\end{corollary}
\todo{This can be improved to $n - \Omega(1)$ using spectral argument, I think.}

\begin{remark}\label{rem:extractor-vs-concentration}
We note that the extractor argument of Lemma~\ref{lem:rand-sketch-extractors} is often much weaker than the arguments we give in Part 2 and Part 3 Theorem~\ref{thm:linear-sketch-uniform} and wouldn't suffice for our applications in Section~\ref{sec:applications}.
In fact, the extractor argument is too weak even for the majority function $Maj_n$.
If the first $100\sqrt{n}$ variables of $Maj_n$ are fixed to $0$ then the resulting restriction has value $0$ with probability $1 - e^{-\Omega(n)}$. Hence for constant error $Maj_n$ isn't an extractor for dimension greater than $100 \sqrt{n}$.
However, as shown in Section~\ref{sec:symmetric} for constant error $\ftwo$-sketch complexity of $Maj_n$ is linear.
\end{remark}

\subsection{Existential lower bound for arbitrary distributions}

Now we are ready to show that an analog of Part 1 of Theorem~\ref{thm:linear-sketch-uniform} doesn't hold for arbitrary distributions, i.e. concentration on a low-dimensional linear subspace doesn't imply existence of randomized linear sketches of small dimension.

\begin{lemma}\label{lem:rand-sketch-complexity-converse}
	For any fixed constant $\epsilon > 0$ there exists a function $f \colon \ftwo^n \rightarrow \oo$ such that $\rl{\epsilon/8}(f) \ge n - 3 \log n$ such that $f$ is $(1 - 2\epsilon)$-concentrated on the $0$-dimensional linear subspace.
\end{lemma}
\begin{proof}
	The proof is based on probabilistic method. Consider a distribution over functions from $\ftwo^n$ to $\oo$ which independently assigns to each $x$ value $1$ with probability $1 - \epsilon/4$ and value $-1$ with probability $\epsilon/4$.
	By a Chernoff bound with probability $e^{- \Omega(\epsilon 2^n)}$ a random function $f$ drawn from this distribution has at least an $\epsilon/2$-fraction of $-1$ values and hence
	$\hat f(\emptyset) = \frac{1}{2^n}\sum_{\alpha \in \ftwo^n} f(x) \ge 1 - \epsilon$.
	This implies that $\hat f(\emptyset)^2 \ge (1 - \epsilon)^2 \ge 1 - 2 \epsilon$ so $f$ is $(1 - 2 \epsilon)$-concentrated on a linear subspace of dimension $0$. However, as we show below the randomized sketching complexity of some functions in the support of this distribution is large.
	
	The total number of affine subspaces of codimension $d$ is at most $(2 \cdot 2^n)^d =2^{(n + 1) d}$ since each such subspace can be specified by $d$ vectors in $\ftwo^n$ and a vector in $\ftwo^d$.
	The number of vectors in each such affine subspace is $2^{n - d}$.
	The probability that less than $\epsilon/8$ fraction of inputs in a fixed subspace have value $-1$ is by a Chernoff bound at most $e^{- \Omega(\eps 2^{n - d})}$.
	By a union bound the probability that a random function takes value $-1$ on less than $\epsilon/8$ fraction of the inputs in any affine subspace of codimension $d$ is at most $e^{- \Omega(\eps 2^{n - d})} 2^{(n + 1) d}$.
	For $d \le n - 3 \log n$ this probability is less than $e^{-\Omega(\epsilon n)}$.
	By a union bound, the probability that a random function is either not an $\epsilon/8$-extractor or isn't $(1 - 2\epsilon)$-concentrated on $\hat f(\emptyset)$ is at most $e^{-\Omega(\epsilon n)} + e^{-\Omega(\epsilon 2^n)} \ll 1$.
	Thus, there exists a function $f$ in the support of our distribution which is an $\epsilon/8$-extractor for any affine subspace of dimension at least $3 \log n$ while at the same time is $(1 - 2 \epsilon)$-concentrated on a linear subspace of dimension $0$. By Lemma~\ref{lem:rand-sketch-extractors} there is no randomized linear sketch of dimension less than $n - 3 \log n$ for $f$ which errs with probability less than $\epsilon/8$.
\end{proof}

%However, despite the fact that the direct converse of Lemma~\ref{lem:rand-sketch-complexity} doesn't hold, we show that for small values of $\epsilon$ a weak converse does in fact hold.
%
%\begin{lemma}
%	If $f$ is $(1 - \epsilon)$-concentrated on a linear subspace of dimension $d$ then there is a randomized linear sketch of dimension $n + d - \log 1/\epsilon + 2 \log 1/\delta + 2$  that predicts $f$ with probability at least $1 - \delta$ for any $\delta > 0$.
%\end{lemma}
%\begin{proof}
%	Let $\mathcal A_d$ be the linear subspace that contains $(1-\epsilon)$ fraction of the Fourier spectrum of $f$.
%	Let $g = \sum_{S \in \mathcal A_d} \hat f(S) \chi_S$ and $h = \sum_{S \notin \mathcal A_d} \hat f(S) \chi_S$.
%	Then $f  = g + h$ where $||g||^2_2 \ge 1 - \epsilon$ and $||h||^2_2 \le \epsilon$.
%	Let $G = sign(g)$. Note $G$ is a $d$-dimensional function since it is a function of $g$.
%	We will need the following standard fact (see e.g.~\cite{GOSSW11,OD14}). We include the proof for completeness.
%	\begin{fact}[~\cite{GOSSW11}]
%		$
%		\Pr_x[f(x) \neq G(x)] \le \epsilon
%		$
%	\end{fact}
%	\begin{proof}
%		We have $\Pr_x[f(x) \neq G(x)] \le \mathbb E_x[|f(x) - g(x)|^2] = ||f - g||^2_2 \le \epsilon$.
%	\end{proof}
%	Thus, for all $x$ we have $f(x) = G(x)H(x)$ where $\Pr_x[H(x) = -1] \le \epsilon$.
%	We can sketch $G$ on every input using $d$ deterministic sketches. Thus, if we can sketch $H$ then we also have a sketch for $f$.

\subsection{Random $\ftwo$-sketching}\label{sec:rand-sketching}

The following result is folklore as it corresponds to multiple instances of the communication protocol for the equality function~\cite{KN97,GKW04} and can be found e.g. in~\cite{MO09} (Proposition 11). We give a proof for completeness.

	\begin{fact}\label{prop:l0-bound}
		A function $f:\ftwo^n \to \ftwo$ such that $\min_{z \in \{0,1\}} \Pr_x[f(x) = z] \le \epsilon$ satisfies
		$$\rl{\delta}(f) \le \log \frac{\epsilon 2^{n+1}}{\delta} .$$
	\end{fact}
	\begin{proof}
		We assume that $\argmin_{z \in \{0,1\}} \Pr_x[f(x) = z] = 1$ as the other case is symmetric.
		Let $T = \{ x \in \ftwo^n | f(x) = 1\}$.
		For every two inputs $x \neq x' \in T$ a random $\ftwo$-sketch $\chi_{\alpha}$ for $\alpha \sim U(\ftwo^n)$ satisfies $\Pr[\chi_\alpha(x) \neq \chi_\alpha(x')] = 1/2$.
		If we draw $t$ such sketches $\chi_{\alpha_1}, \dots, \chi_{\alpha_t}$ then $\Pr[\chi_{\alpha_i}(x) = \chi_{\alpha_i}(x'), \forall i \in [t] ] = 1/2^t$.
		For any fixed $x \in T$ we have:
		\begin{align*}
		\Pr[\exists x'\neq x \in T \text{ } \forall i \in [t]: \chi_{\alpha_i}(x) = \chi_{\alpha_i}(x') ] \le  \frac{|T| - 1}{2^t} \le \frac{\epsilon 2^n}{2^t} \le \frac{\delta}{2}.
		\end{align*}
		Conditioned on the negation of the event above for a fixed $x \in T$ the domain of $f$ is partitioned by the linear sketches into affine subspaces such that $x$ is the only element of $T$ in the subspace that contains it.
		We only need to ensure that we can sketch $f$ on this subspace which we denote as $\mathcal A$.
		On this subspace $f$ is isomorphic to an OR function (up to taking negations of some of the variables) and hence can be sketched using $O(\log 1/\delta)$ uniformly random sketches with probability $1 - \delta/2$.
		For the OR-function existence of the desired protocol is clear since we just need to verify whether there exists at least one coordinate of the input that is set to $1$.
		In case it does exist a random sketch contains this coordinate with probability $1/2$ and hence evaluates to $1$ with probability at least $1/4$. Repeating $O(\log 1/\delta)$ times the desired guarantee follows.
	\end{proof}

\section{Tightness of Theorem~\ref{thm:linear-sketch-uniform} for the Majority function} \label{app:tightness}
An important question is whether Part 3 of Theorem~\ref{thm:linear-sketch-uniform} is tight. In particular, one might ask whether the dependence on the error probability can be improved by replacing $\dgap{d}{f}$ with a larger quantity.
As we show below this is not the case and hence Part 3 of Theorem~\ref{thm:linear-sketch-uniform} is tight. 

We consider the majority function $Maj_n$ where $n$ is an odd number.
The total Fourier weight on Fourier coefficients corresponding vectors of Hamming weight $k$ is denoted as $W^k(f) = \sum_{\alpha \colon \|\alpha\|_0 = k} \hat f(\alpha)^2$.
For the majority function it is well-known (see e.g.~\cite{OD14}) that for $\xi = \left(\frac{2}{\pi}\right)^{3/2}$ and odd $k$ it holds that: 
$$W^k(Maj_n) = \xi k^{-3/2} (1 \pm O(1/k)).$$
Since $Maj_n$ is a symmetric function whose spectrum decreases monotonically with the Hamming weight of the corresponding Fourier coefficient by a normalization argument as in Lemma~\ref{lem:linear-subspace-hamming-weight-intersection} among all linear subspaces of dimension $d$ the maximum Fourier weight is achieved by the standard subspace $\mathcal S_d$ which spans $d$ unit vectors.
Computing the Fourier weight of $\mathcal S_{n-1}$ we have:
\begin{align*}
\sum_{\alpha \in \mathcal S_{n-1}} \widehat{Maj}_n(\alpha)^2 &= 1 - \sum_{\alpha \notin \mathcal S_{n-1}} \widehat{Maj}_n(\alpha)^2 \\
&= 1 - \sum_{i = 0}^{n/2 - 1} W^{2i + 1}(Maj_n) \frac{\binom{n - 1}{2i} }{ \binom{n}{2i + 1}} \\
&= 1 - \sum_{i = 0}^{n/2 - 1} \xi \frac{1}{(2i + 1)^{3/2}} \frac{2i + 1 }{n} \left(1 \pm O\left(\frac{1}{2i + 1}\right)\right) \\
&= 1 - \frac{\gamma}{\sqrt{n}} \pm O\left(\frac{1}{n^{3/2}}\right),
\end{align*}
where $\gamma > 0$ is an absolute constant.
Thus, we can set $\epsu{n}{Maj_n} = 1, \epsu{n-1}{Maj_n} = 1 - \frac{\gamma}{\sqrt{n}} - O(1/n^{3/2})$ in Part 3 of Theorem~\ref{thm:linear-sketch-uniform}. This gives the following corollary:
\begin{corollary}
	It holds that $\distcu{\delta}(Maj_n^+) \ge n$, where $\delta = \frac{\gamma}{\sqrt{n}} + O\left(\frac{1}{n^{3/2}}\right)$ for some constant $\gamma >0$.
\end{corollary}

Tightness follows from the fact that error $O(1/\sqrt{n})$ for $Maj_n$ can be achieved using a trivial $(n-1)$-bit protocol in which Alice sends the first $n-1$ bits of her input $x_1, \dots, x_{n-1}$ and Bob outputs $Maj_{n - 1}(x_1 + y_1, x_2 + y_2, \dots, x_{n -1} + y_{n - 1})$.
The only inputs on which this protocol can make an error are inputs where there is an equal number of zeros and ones among $x_1 + y_1, \dots, x_{n-1} + y_{n-1}$. It follows from the standard approximation of binomials that such inputs are an $O(1/\sqrt{n})$ fraction under the uniform distribution.

\end{document}